\documentclass[11pt,fleqn]{article}

\usepackage{amsmath,amssymb,amsthm,enumerate}

\allowdisplaybreaks

\newcommand{\p}{\partial}
\newcommand{\const}{\mathop{\rm const}\nolimits}
\newcommand{\Equiv}{\mathop{ \sim}}

\newcommand{\sign}{\mathop{\rm sign}\nolimits}
\newcommand{\thetbn}{\arabic{nomer}}

\newcommand{\prho}{p}

\newcommand{\todo}[1][\null]{\ensuremath{\clubsuit}}
\newcommand{\checkit}[2][\null]{\ensuremath{\clubsuit}\marginpar{\raggedright\mathversion{bold}\footnotesize\bf #1}#2\ensuremath{\clubsuit}}

\newcommand{\noprint}[1]{}

\newcounter{tbn}
\newcounter{casetran}
\newcounter{mcasenum}

\newtheorem{theorem}{Theorem}
\newtheorem{lemma}{Lemma}

{\theoremstyle{definition}

\newtheorem{note}{Note}
}

\begin{document}

\par\noindent {\LARGE\bf
Group Analysis of
Variable Coefficient\\ Diffusion--Convection Equations.\\ I. Enhanced Group Classification
\par}
{\vspace{4mm}\par\noindent {\bf N.M. Ivanova~$^\dag$, R.O. Popovych~$^\ddag$ and C. Sophocleous~$^\S$
} \par\vspace{2mm}\par}
{\vspace{2mm}\par\noindent {\it
$^\dag{}^\ddag$~Institute of Mathematics of NAS of Ukraine, 
3 Tereshchenkivska Str., 01601 Kyiv, Ukraine\\
}}
{\noindent \vspace{2mm}{\it
$\phantom{^\dag{}^\ddag}$~e-mail: ivanova@imath.kiev.ua, rop@imath.kiev.ua
}\par}

{\par\noindent\vspace{2mm} {\it
$^\ddag$~Fakult\"at f\"ur Mathematik, Universit\"at Wien, Nordbergstra{\ss}e 15, A-1090 Wien, Austria
} \par}

{\vspace{2mm}\par\noindent {\it
$^\dag{}^\S$~Department of Mathematics and Statistics, University of Cyprus,
CY 1678 Nicosia, Cyprus\\
}}
{\noindent {\it
$\phantom{^\dag{}^\S}$~e-mail: christod@ucy.ac.cy
} \par}

{\vspace{7mm}\par\noindent\hspace*{8mm}\parbox{140mm}{\small
We discuss the classical statement of group classification problem
and some its extensions in the general case. After that,
we carry out the complete extended group classification for a class of
(1+1)-dimensional nonlinear diffusion--convection equations
with coefficients depending on the space variable.
At first, we construct the usual equivalence group and the extended one including
transformations which are nonlocal with respect to arbitrary elements.
The extended equivalence group has interesting structure since it contains a non-trivial subgroup of
non-local gauge equivalence transformations.
The complete group classification of the class under consideration is carried out
with respect to the extended equivalence group
and with respect to the set of all point transformations.
Usage of extended equivalence and correct choice of gauges of arbitrary elements
play the major role for simple and clear formulation of the final results.
The set of admissible transformations of this class is preliminary investigated.
}\par\vspace{7mm}}

\section{Introduction}

In the present series of papers we perform extended symmetry analysis of a class of variable coefficient
nonlinear diffusion--convection equations of the general form
\begin{equation} \label{eqDKfgh}
f(x)u_t=(g(x)A(u)u_x)_x+h(x)B(u)u_x,
\end{equation}
where $f=f(x),$ $g=g(x),$ $h=h(x),$ $A=A(u)$ and $B=B(u)$ are arbitrary smooth functions of their variables,
$fghA\ne0$ and $(A_u,B_u)\ne(0,0)$.
If $B=0$, it is convenient to assume $h=1$ for determinacy.
Different kinds of properties and objects related to the framework of symmetry approach such as
Lie, nonclassical and potential symmetries, equivalence groups, admissible transformations, exact solutions and
conservation laws are investigated for these equations.

Equations from class~(\ref{eqDKfgh}) can be used to model a wide variety of phenomena in physics,
chemistry, mathematical biology etc.
Constant coefficient equations of the form~(\ref{eqDKfgh}), often called the Richard's equations, describe
vertical one-dimensional transport of water in a homogeneous non-deformable porous media.
This equation arises naturally in certain physical applications.
Thus, for example, superdiffusivities of this type have been proposed in~\cite{Gennes1983}
as a model for long-range Van der Waals interactions in thin
films spreading on solid surfaces. Equation~(\ref{eqDKfgh}) also appears in the study of
cellular automata and interacting particle systems with self-organized criticality
(see \cite{Chayes&Osher&Ralston1993} and references therein).
It is a model of water flow in unsaturated soil~\cite{Richard's1931}.
When $B(u)=0$ equation~(\ref{eqDKfgh})
describes stationary motion of a boundary layer of a fluid over a flat
plate, vortex of an incompressible fluid in porous medium for
polytropic relations of gas density and pressure~\cite{Barenblatt1952}.
The outstanding representative of the class of equations (\ref{eqDKfgh}) is the Burgers
equation that is the mathematical model of a large number of physical phenomena.

The Lie symmetry transformations of linear equations of form~(\ref{eqDKfgh}) ($A,B=\const$)
were determined by Lie~\cite{Lie1881}
in his classification of linear second-order PDEs with two independent variables.
(See also a modern treatment of this subject in~\cite{Ovsiannikov1982}.)
Investigation of the nonlinear equations from class~(\ref{eqDKfgh})
by means of symmetry methods were started in 1959 with the paper~\cite{Ovsiannikov1959} by Ovsiannikov
where he studied Lie symmetries of the nonlinear diffusion equations $u_t=(A(u)u_x)_x$.
The next considered class~\cite{Katkov1965} covers the generalized Burgers equations $u_t=u_{xx}+B(u)u_x$.
Dorodnitsyn~\cite{Dorodnitsyn1982} carried out the group classification of the class $u_t=(A(u)u_x)_x+C(u)$.
Akhatov, Gazizov and Ibragimov~\cite{Akhatov&Gazizov&Ibragimov1987} classified the equations $u_t=A(u_x)u_{xx}$.
This latter classification gave quasi-local symmetries of the nonlinear diffusion equations $u_t=(A(u)u_x)_x$.
A number of authors~\cite{Oron&Rosenau1986,Edwards1994,Yung&Verburg&Baveye1994}
investigated Lie symmetries of the class of the constant-coefficient diffusion--convection equations
\[
u_t=(A(u)u_x)_x+B(u)u_x.
\]
However, its complete and strong group classification was obtained only recently~\cite{Popovych&Ivanova2004NVCDCEs}.
M.L.~Gandarias~\cite{Gandarias1996} supplemented this list considering Lie symmetries of
the variable-coefficient equation
$
u_t=(u^n)_{xx}+g(x)u^m+f(x)u^su_x.
$
These  results 
were generalized in~\cite{Cherniga&Serov1998}, where Lie symmetries of nonlinear reaction--diffusion equations
with convection term $u_t=(A(u)u_x)_x+B(u)u_x+C(u)$ were studied.
The complete group classification of the variable-coefficient diffusion--con\-vec\-tion equations
\[
f(x)u_t=(g(x)A(u)u_x)_x+B(u)u_x,
\]
was carried out in~\cite{Popovych&Ivanova2004NVCDCEs}.
In the recent papers~\cite{Vaneeva&Johnpillai&Popovych&Sophocleous2006,Vaneeva&Popovych&Sophocleous2009}
the complete group classification of the class of reaction--diffusion equations $f(x)u_t=(g(x)u^nu_x)_x+h(x)u^m$
were presented.
In~\cite{Ivanova&Sophocleous2006} we classified Lie symmetries of equations from class~\eqref{eqDKfgh}
with respect to its usual equivalence group.
It should be noted that all the above mentioned equations are particular
cases of the more general class of equations
\[
u_t=F(t,x,u,u_x)u_{xx}+G(t,x,u,u_x),
\]
that was classified in~\cite{Basarab-Horwath&Lahno&Zhdanov2001}.
However, since the equivalence group of the latter class
is essentially wider than those for the considered subclasses,
the results of \cite{Basarab-Horwath&Lahno&Zhdanov2001}
cannot be directly used to symmetry classification of the above mentioned equations.
Nevertheless, these results are useful to find additional equivalence transformations for the above classes.

Equations of form~(\ref{eqDKfgh}) have been also investigated from points of view of other kinds of symmetries.
For instance, potential symmetries of subclasses of~(\ref{eqDKfgh}), where e.g.\ either $f=g=h=1$ or $B=0$,
were studied in \cite{Bluman&Kumei1989,Bluman&Reid&Kumei1988,Ivanova&Popovych&Sophocleous&Vaneeva2009,
Popovych&Ivanova2003PETs,Sophocleous1996,Sophocleous2000,Sophocleous2003}.
Conditional symmetries (usual and potential ones) for some subclasses of~(\ref{eqDKfgh}), in particular for the fast diffusion equation,
were constructed in~\cite{Bluman&Yan2005,Gandarias2001,Popovych&Vaneeva&Ivanova2005}.

Study of group classification problems is interesting not only from the purely
mathematical point of view, but is also important for applications~\cite{Ovsiannikov1982}.\
Physical models are often constrained with a priori requirements to symmetry properties
following from physical laws, for example, from the Galilean or special relativity principals.\
Moreover, modelling differential equations could contain parameters or functions
which have been found experimentally and so are not strictly fixed.\
(These parameters and functions are called arbitrary elements.)\
At the same time, mathematical models should be enough simple to analyze and solve them.
Symmetry approach allows us to take the following relevancy criterion for choosing parameter values.
Modelling differential equations have to admit a symmetry group with certain properties
or the most extensive symmetry group from the possible ones.\
This directly leads to necessity of solving group classification problems.

The group classification in a class of differential equations is reduced to integration of a  complicated overdetermined system
of partial differential equations with respect to both coefficients of infinitesimal symmetry operators and arbitrary elements.
This is why it is a much more complicated problem than finding the Lie symmetry group of a single system of differential equation.
Whereas programs for solving the latter problem had been created for the most of existing symbolic calculations packages, the significant progress
in computer realization of the group classification algorithm was achieved only recently~\cite{Wittkopf2004PhD}.
There were a lot of attempts in computation of Lie symmetries of differential equations using different systems of computer algebras,
such as MATHEMATICA, MAPLE, MACSYMA, REDUCE, AXIOM, MuPAD etc.
Different symbolic calculation packages~\cite{Cheviakov2006,Dimas&Tsubelis2004,Head1993,Wolf1992}
(see also detailed review in~\cite{Hereman1994,Hereman1997})
created within the above mentioned systems are able to construct determining equations for symmetry operators
and integrate them in simple cases.
However, the existing packages have a number of essential disadvantages
(e.g., restrictions on nonlinearities) and cannot perform the correct and exhaustive group classification
in quite complicated classes of differential equations.

In general, problems of group classification, except for really trivial cases, are very difficult.
This can be illustrated by a multitude of papers
where classification problems are solved incorrectly or incompletely.
There are also many papers on ``preliminary group classification'' where authors list some cases
of Lie symmetry extension but do not claim that classification problems are solved completely.
For this reason finding an effective approach to simplification of a group classification problem
is essentially equivalent to feasibility of solving the problem at all.

The ultimate goal of the first part of this paper series is to present a modern treatment of group classification
and to give an example of exhaustive solution of a group classification problem in quite a difficult case.
After discussing this problem and some its possible extensions in the general setting,
we carry out the enhanced group classification for class~\eqref{eqDKfgh}.
Different original tools are applied that allows us to solve the problem completely.

Thus, we construct both the usual equivalence group~$G^{\sim}$ and the \emph{extended} one~$\hat G^{\sim}$
including transformations which are nonlocal with respect to arbitrary elements.
The structure of the extended equivalence group~$\hat G^{\sim}$ is investigated.
It has a non-trivial subgroup of (nonlocal) \emph{gauge equivalence transformations}.
Acting by these transformations is equivalent to rewriting equations in another form.
In spite of real equivalence transformations, a role of gauge ones in group classification
comes not to choice of representatives in sets of equivalent equations but to choice of form of these representatives.
Application of the nonlocal gauge transformations is important under solving the group classification problem
in class~\eqref{eqDKfgh}.
Moreover, it seems practically impossible to obtain closed classification results for the class under consideration
without using `unusual' transformations from the extended equivalence group.
This conclusion was made after the preliminary group classification of class~\eqref{eqDKfgh}
with respect to the usual equivalence group~$G^{\sim}$~\cite{Ivanova&Popovych&Sophocleous2004}.
The corresponding list of equations is not optimal since for some equations
it includes a number of their different representatives.

Another substantial improvement proposed in this paper to group classification methods is
the \emph{technique of variable gauges} of arbitrary elements by equivalence transformations.
Two gauges $g=1$ and $g=h$ are used simultaneously.
Both of them are natural generalizations of the gauge $g=1$
effectively applied in~\cite{Popovych&Ivanova2004NVCDCEs} for the subclass $h=1$ of~\eqref{eqDKfgh}
although the gauge $g=h$ seems more complicated than the gauge $g=1$.
Indeed, we present complete group classifications in both the gauges to compare the obtained results.
The comparison shows that the best way for gauging arbitrary elements of class~\eqref{eqDKfgh} is
to combine the gauges depending on classification cases.
For optimality of presentation, the classification list enhanced with the variable gauge will be
given in the second part of the series~\cite{Ivanova&Popovych&Sophocleous2006Part2}
where it will be utilized for the construction of exact solutions of equations from class~\eqref{eqDKfgh}
via Lie reductions.

The method of \emph{furcate split} is also used under the classification.
It is a simple and effective tool based on a specific way of working with classifying conditions without
their complete integration.
First the method was proposed in~\cite{Nikitin&Popovych2001} and then
applied in different modifications to solve a number of group classification problems
\cite{Boyko&Popovych2001,Ivanova&Sophocleous2006,Popovych&Cherniha2001,Popovych&Ivanova2004NVCDCEs,Vasilenko&Yehorchenko2001}.
Its application is not crucial for the classification problem of class~\eqref{eqDKfgh} but
the achieved simplifications of the proof are considerable.

The last but not least tool is the regular usage of \emph{additional equivalences} between cases
of Lie symmetry extension, which are not generated by transformations from the (extended) equivalence group.
Nontrivial additional equivalence transformations were first constructed in~\cite{Kurdyumov&Posashkov&Sinilo1989}
for special cases of (constant coefficient) nonlinear reaction--diffusion equations.
(The~group classification of these equations was carried out in~\cite{Dorodnitsyn1982}.
See also~\cite[Chapter~10]{Ibragimov1994V1}.)
In this paper all the additional equivalences in class~\eqref{eqDKfgh} are constructed and explicitly presented.
The usage of these equivalences gives an auxiliary test for the classification results and
simplifies further application of them, e.g., to finding exact solutions.
Additional equivalences are naturally embedded in the framework of \emph{admissible transformations}
that gives a constructive way for the exhaustive description of them.
The structure of the set of admissible transformations of class~\eqref{eqDKfgh} is too complicated.
In this paper only preliminary results about this set, sufficient for the description of additional equivalences,
are presented.

The rest of the paper is organized as follows.
First of all (Section~\ref{SectionOnGrClMethod})
we describe the general formulation of group classification problems and an algorithm of their solution.
In Section~\ref{SectionAdmisTransfTheory} we adduce some basic notions on admissible transformations
and different generalizations of equivalence groups of classes of differential equations.
Then (Section~\ref{SectionOnEquivTransOfEqDKfgh}) the complete group of usual equivalence transformations
for class~\eqref{eqDKfgh} and the extended one including transformations which are nonlocal
with respect to arbitrary elements are found.
Taking the non-trivial subgroup of gauge equivalence transformations into account strongly
simplifies the solution of the group classification problem.
The results of group classification for class~\eqref{eqDKfgh} in two different gauges
($g=1$ and $g=h$) are presented and analyzed in Section~\ref{SectionOnGrClOfDCEfgh}.
We note that for both gauges two essentially different classifications are presented:
the classification with respect to the (extended) equivalence group and
the classification with respect to all possible point transformations.
The sketch of the proof of the obtained results is given in Section~\ref{SectionProofOfGrCl}.
Section~\ref{SectionAdmisTransf} is devoted to the preliminary investigation of admissible transformations
for equations from class~\eqref{eqDKfgh}.

The derived results form a basis of the further consideration of class~\eqref{eqDKfgh} from the symmetry point of view.
An important notion of contractions of differential equations will be introduced and developed
in the second part~\cite{Ivanova&Popovych&Sophocleous2006Part2} of the series.
Contractions between cases of Lie symmetry extensions in class~\eqref{eqDKfgh} will be found.
Lie and non-Lie reductions and exact solutions of equations from these class also will be constructed.
In the third part~\cite{Ivanova&Popovych&Sophocleous2006Part3} we study the local and potential conservation
laws of equations from class~\eqref{eqDKfgh} and find all possible inequivalent potential systems of equations from class~\eqref{eqDKfgh}.
The ``contraction concept" will be further generalized via introducing the notion of contractions of conservation laws.
In the last, fourth, part~\cite{Ivanova&Popovych&Sophocleous2006Part4} of the series we investigate symmetry properties of potential systems
associated with equations~\eqref{eqDKfgh}
and describe connections between the potential and point symmetries of diffusion--convection equations.

\section{Classical algorithm of group classification and\\ additional equivalences}\label{SectionOnGrClMethod}

Group classification is one of symmetry tools used to choose physically relevant models
from parametric classes of systems of (ordinary or partial) differential equations.
The parameters can be constants or functions of independent variables, unknown functions and their derivatives.

Since in the literature there exist different points of view on the framework of group classification,
below we give a modern exposition of the classical Lie--Ovsiannikov approach~\cite{Ovsiannikov1982} in detail and then
briefly formulate its possible modifications with extensions of equivalences under consideration or
with investigation of wider classes of symmetries.

Consider the class~$\mathcal L|_{\mathcal S}$ of systems~$\mathcal L^\theta$: $L(x,u_{(\prho)},\theta(x,u_{(\prho)}))=0$
of $l$~differential equations
for $m$~unknown functions $u=(u^1,\ldots,u^m)$
of $n$~independent variables $x=(x_1,\ldots,x_n)$,
which is parameterized with the functions~$\theta=\theta(x,u_{(\prho)})$.
Here $u_{(\prho)}$ denotes the set of all the derivatives of~$u$ with respect to $x$
of order not greater than~$\prho$, including $u$ as the derivatives of the zero order.
$L=(L^1,\ldots,L^l)$ is a tuple of $l$ fixed functions depending on $x,$ $u_{(\prho)}$ and $\theta$.
$\theta$ denotes the tuple of arbitrary (parametric) functions
$\theta(x,u_{(n)})=(\theta^1(x,u_{(\prho)}),\ldots,\theta^k(x,u_{(\prho)}))$
running through the set~${\cal S}$ of solutions of the system
$S(x,u_{(\prho)},\theta_{(q)}(x,u_{(\prho)}))=0.$
This system consists of differential equations with respect to $\theta$,
where $x$ and $u_{(\prho)}$ play the role of independent variables
and $\theta_{(q)}$ stands for the set of all the partial derivatives of $\theta$ of order not greater than $q$.
Usually the set $\mathcal S$ is additionally constrained by the non-vanish condition
$\Sigma(x,u_{( p)},\theta_{(q)}(x,u_{( p)}))\ne0$ with a differential function~$\Sigma$.
In what follows we call the functions $\theta$ {\it arbitrary elements} of the class~$\mathcal L|_{\mathcal S}$.

Let \smash{$\mathcal L^\theta_{(k)}$}
denote the set of all algebraically independent differential consequences
of the system~$\mathcal L^\theta$ that have, as differential equations, orders not greater than $k$.
Under the local approach, the system~\smash{$\mathcal L^\theta_{(k)}$} is identified with the manifold
determined by~\smash{$\mathcal L^\theta_{(k)}$} in the jet space~$J^{(k)}$.

Each one-parametric group of point transformations
that leaves the system~$\mathcal L^\theta$ invariant corresponds to an
infinitesimal symmetry operator of the form
\[
Q=\xi^i(x,u)\p_{x_i}+\eta^a(x,u)\p_{u^a}.
\]
Here and below the summation over the repeated indices is assumed.
The indices~$i$ and~$a$ run from~1 to~$n$ and from~1 to~$m$, respectively.

The complete set of such groups generates the \emph{maximal Lie invariance} (or \emph{principal}) \emph{group}
$G^\theta=G^{\max}(\mathcal L^\theta)$ of the system~$\mathcal L^\theta$.
The group $G^\theta$ is the connected component of the unity in the group $\bar G^\theta$ of all point symmetry
transformations of~$\mathcal L^\theta$.
A subgroup $G^\theta_\mathrm{d}$ of $\bar G^\theta$, which can be canonically identified with
$\bar G^\theta/G^\theta$, is called the group of \emph{discrete} symmetries in contrast to the group $G^\theta$ of
\emph{continuous} symmetries.
The \emph{maximal Lie invariance} (or \emph{principal}) \emph{algebra} $A^\theta=A^{\max}(\mathcal L^\theta)$
formed by the infinitesimal symmetry operators of~$\mathcal L^\theta$ is the Lie algebra of the principal group $G^\theta$.

The \emph{infinitesimal invariance criterion} of the system~$\mathcal L^\theta$
with respect to the Lie symmetry operator~$Q$ has the form~\cite{Ovsiannikov1982,Olver1986},
\[
Q_{(p)} L(x,u_{(\prho)},\theta(x,u_{(\prho)}))\big|_{\mathcal L^\theta_{(p)}}=0, \qquad
Q_{(p)}:=Q+\sum_{0<|\alpha|{}\leqslant  p} \eta^{a\alpha}\p_{u^a_\alpha},
\]
i.e., the result of acting by $Q_{(p)}$ on $L$ vanishes on the manifold~$\mathcal L^\theta_{(p)}$.
Here $Q_{(p)}$ denotes the standard $p$-th prolongation of the operator~$Q$,
$\eta^{a\alpha}=D_1^{\alpha_1}\ldots D_n^{\alpha_n}Q[u^a]+\xi^iu^a_{\alpha,i}$,
$ D_i=\p_i+u^a_{\alpha,i}\p_{u^a_\alpha}$
is the operator of total differentiation with respect to the variable~$x_i$,
$Q[u^a]=\eta^a(x,u)-\xi^i(x,u)u^a_i$ is the characteristic of operator $Q$, associated with~$u^a$.
The tuple $\alpha=(\alpha_1,\ldots,\alpha_n)$ is a multi-index,
$\alpha_i\in\mathbb{N}\cup\{0\}$, $|\alpha|\mbox{:}=\alpha_1+\cdots+\alpha_n$.
The variables $u^a_\alpha$ and $u^a_{\alpha,i}$ of the jet space $J^{(r)}$ correspond to the derivatives
\[
\frac{\p^{|\alpha|}u^a}{\p x_1^{\alpha_1}\ldots\p x_n^{\alpha_n}}
\quad\mbox{and}\quad
\frac{\p^{|\alpha|+1}u^a}
{\p x_1^{\alpha_1}\ldots\p x_{i-1}^{\alpha_{i-1}}\p x_i^{\alpha_i+1}\p x_{i+1}^{\alpha_{i+1}}\ldots\p x_n^{\alpha_n}}.
\]

The \emph{kernel} of principal groups of the class~$\mathcal L|_{\mathcal S}$ is the group
$G^{\cap}=G^{\cap}(\mathcal L|_{\mathcal S})=\bigcap_{\theta\in{\cal S}}G^\theta$.
Its~Lie algebra is the \emph{kernel of principal algebras} (or, the kernel algebra)
$A^\cap=A^\cap(\mathcal L|_{\mathcal S})=\bigcap_{\theta\in{\cal S}}A^\theta$.

The group of point transformations in the space of the variables $(x,u,\theta)$,
which preserve the form of the systems from~$\mathcal L|_{\mathcal S}$
(i.e., transforming any system from~$\mathcal L|_{\mathcal S}$ to the system from the same class)
is called the \emph{equivalence group} of the class~$\mathcal L|_{\mathcal S}$ and is denoted by
$G^{\Equiv}=G^{\Equiv}(\mathcal L|_{\mathcal S}).$
The conventional direct method is the most effective for calculating the entire equivalence group.
Sometimes a subgroup of~$G^{\Equiv}$ is considered instead of the complete (point) equivalence group.
For example, it can be the group of continuous equivalence transformations
which can be calculated by the simpler infinitesimal method adopted to the case of equivalence transformations.

Different extensions of the standard notion of equivalence group are also possible~\cite{Popovych2006}.
Usually one considers only transformations
being projectible on the space of the variables~$x$ and $u$.
Let us remind that
the point transformation~$\varphi$: $\tilde z=\varphi(z)$ in the space of the variables~$z=(z_1,\ldots,z_k)$
is called {\em projectible} on the space of variables~$z'=(z_{i_1},\ldots,z_{i_{k'}})$, where $1\le i_1<\cdots<i_{k'}\le k$,
if the expressions for~$\tilde z'$ depend only on~$z'$.
We denote the restriction of~$\varphi$ on the space of~$z'$ as $\varphi|_{z'}$: $\tilde z'=\varphi|_{z'}(z')$.
If the arbitrary elements~$\theta$ explicitly depend on $x$ and $u$ only (one always can do it,
assuming derivatives as new dependent variables), we may consider
the generalized equivalence group~$G^{\Equiv}_{\rm gen}(\mathcal L|_{\mathcal S})$~\cite{Meleshko1994},
admitting dependence of transformations of~$(x,u)$ on~$\theta$.

Sometimes it is possible to consider other generalizations of equivalence groups, e.g.,
groups of transformations involving nonlocality with respect to arbitrary elements
(see~\cite{Ivanova&Popovych&Sophocleous2004} and Section~\ref{SectionOnEquivTransOfEqDKfgh} of the present paper).

The \emph{problem of group classification} is to find all possible
inequivalent cases of extensions of $A^{\max}$, i.e.,
to list all $G^{\Equiv}$-inequivalent  values of the arbitrary parameters~$\theta$
satisfying the condition $A^\theta\ne A^\cap$.
The exhaustive investigation of this problem in the classical statement involves
finding the kernel $A^\cap$ of the principal algebras, the construction of the equivalence group $G^{\Equiv}$
and the description of all possible $G^{\Equiv}$-inequivalent values of parameters~$\theta$ that admit
the principal algebras wider than $A^\cap$.
The corresponding procedure is presented in the form
of the following algorithm~\cite{Akhatov&Gazizov&Ibragimov1989,Ovsiannikov1982}.

\smallskip

{\it Step 1. Determining equations.}
The infinitesimal invariance criterion
implies a linear system of differential equations for coefficients of an arbitrary Lie symmetry operator~$Q$
of the system~$\mathcal L^\theta$.

This system is obviously parameterized with the arbitrary elements~$\theta$.
After splitting it with respect to the unconstrained variables (i.e., unconstrained derivatives of~$u$),
we may obtain that some of the determining equations do not contain arbitrary elements
and therefore can be integrated immediately. The others
(i.e., the equations containing arbitrary elements explicitly)
are called the {\em classifying equations}.
The main difficulty of group classification is the need to solve the classifying equations with
respect to the coefficients of the operator $Q$ and arbitrary elements simultaneously.

\smallskip

{\it Step 2. Kernel algebra.}
The decomposition of the determining equations
with respect to all the unconstrained derivatives of arbitrary elements results in
a linear system of partial differential equations for coefficients of the
infinitesimal operator $Q$ only. Solving this system yields the algebra $A^\cap$.

\smallskip

{\it Step 3. Equivalence group.}
In order to construct the equivalence group $G^{\Equiv}$ of the class~$\mathcal L|_S$,
we have to investigate the point symmetry transformations
of the united system~$\mathcal L^\theta$ and~$\mathcal S$, considering it as a system
of partial differential equations with respect to $\theta$ with the
independent variables $x$ and $u_{(p)}$.
The transformation components for $k$th order derivatives of~$u$ are obtained by $k$th order prolongation
of the transformation components for~$u$.
After restricting ourselves in studying the connected component of unity in $G^{\Equiv},$
we can use the Lie infinitesimal method. To find the complete  equivalence group
(including discrete transformations), we have to use the more complicated direct method.

\smallskip

{\it Step 4. Symmetry extensions.}
If $A^{\theta}$ is an extension of $A^\cap$
(i.e., $A^\theta\ne A^\cap$) then
the classifying equations define a system of nontrivial equations for~$\theta$.
Depending on their form and number, we obtain different cases of extensions of~$A^\cap$.
To integrate completely the determining equations, we have to investigate a large number of such cases.
In order to avoid a cumbersome enumeration of possibilities in solving the determining equations,
we can use, for instance, algebraic methods
\cite{Basarab-Horwath&Lahno&Zhdanov2001,Gagnon1989c,Lahno&Spichak&Stognii2002,Zhdanov&Lahno1999},
a method which involves the investigation of compatibility of the classifying
equations \cite{Boyko&Popovych2001,Nikitin&Popovych2001,Popovych&Cherniha2001,Popovych&Ivanova2004NVCDCEs,Vasilenko&Yehorchenko2001}
or combined methods, in particular, methods based on the (strong) normalization property of
a class of systems \cite{Ivanova&Popovych&Eshraghi2005GammaNormalizedSerbia,Popovych2006,
Popovych&Eshraghi2004Norm,Popovych&Ivanova&Eshraghi2004Cubic,Popovych&Ivanova&Eshraghi2004Gamma}.

\smallskip

The latter step can be completed with calculating explicit conditions (namely, systems of differential equations)
on arbitrary elements, providing extensions of Lie symmetry.
(See, e.g., \cite{Borovskikh2004,Borovskikh2006} for examples of such calculations.)
In other words, the subsets of values of arbitrary elements, each of which is formed by the equations
equivalent to a listed inequivalent extension case, should be constructively described.

\medskip

The result of application of the above algorithm  is a list of equations with their
Lie invariance algebras.
The problem of group classification is assumed to be completely solved if
\begin{enumerate}[\it i\rm)]\itemsep=0ex
\item
the list contains all the possible inequivalent cases of extensions;
\item
all the equations from the list are mutually inequivalent with respect to
the transformations from $G^{\Equiv}$ (or the considered generalization of $G^{\Equiv}$);
\item
the obtained algebras are the maximal invariance algebras of the corresponding equations.
\end{enumerate}
The list may include equations being mutually equivalent with respect to
point transformations which do not belong to $G^{\Equiv}.$
Knowing such \emph{additional} equivalences allows ones to essentially simplify
further investigation of~$\mathcal L|_S$.
Constructing them can be considered as the fifth step of the algorithm of group classification.
Then, the above enumeration of requirements to the resulting list of classification can be completed
by the following item:
\begin{enumerate}[\it i\rm)]\setcounter{enumi}{3}
\item
all the possible additional equivalences between the listed equations are constructed in
an explicit form.
\end{enumerate}
A way for finding additional equivalences is based on the fact that similar equations
have similar maximal invariance algebras.
Another way is systematical study of all possible \emph{admissible} transformations.
The precise definition of such transformations is given in the next section.

Besides the above mentioned generalizations of the problem of group classification to classification of Lie symmetries
with respect to different generalized groups of equivalence transformations, one can consider classifications
of different kinds of symmetries (potential, contact, conditional etc.) with respect to different sets of transformations.
For example, in~\cite{Gazizov1987} Gazizov carried out the group classification of contact symmetries
of nonlinear potential filtration equations
with respect to the group of contact equivalence transformations.

Different approaches can be used for classifications of potential symmetries. The most natural one is the classification with
respect to the group of transformations which is (in some sense) a simple extension of the usual (generalized) point equivalence group
to the potential variables~\cite{Bluman&Kumei1989,Bluman&Reid&Kumei1988,Popovych&Ivanova2003PETs,Sophocleous1996,Sophocleous2000,Sophocleous2003}.
Another possibility is to additionally consider purely potential equivalence transformations~\cite{Popovych&Ivanova2003PETs}.

The most common classification problems for the nonclassical (conditional) symmetries~\cite{Bluman&Cole1969} is the one of the classification
of (systems of) differential equations admitting conditional symmetries and classification of conditional symmetries
of a single system with respect to its Lie symmetry group.
(See, e.g.,~\cite{Fushchych&Tsyfra1987,Popovych&Vaneeva&Ivanova2005,Zhdanov&Tsyfra&Popovych1999}
for precise definitions and examples of solving classification problems.)

\section{Admissible transformations, conditional equivalence groups\\ and other generalizations}\label{SectionAdmisTransfTheory}

For $\theta,\tilde\theta\in\mathcal S$ we call the set of point transformations that
map the system~$\mathcal L^\theta$ into the system~$\smash{\mathcal L^{\tilde\theta}}$
the \emph{set of admissible transformations from~$\mathcal L^\theta$ into~$\smash{\mathcal L^{\tilde\theta}}$}
and denote it by $\mathrm{T}(\theta,\tilde\theta)$.
If the systems~$\mathcal L^\theta$ and $\smash{\mathcal L^{\tilde\theta}}$ are equivalent with respect to
point transformations then
$\mathrm{T}(\theta,\tilde\theta)=G^{\theta}\circ\varphi^0=\varphi^0\circ G^{\smash{\tilde\theta}}$,
where $\varphi^0$ is a fixed transformation from~$\mathrm{T}(\theta,\tilde\theta)$.
Otherwise $\mathrm{T}(\theta,\tilde\theta)=\varnothing$.
The set
$\mathrm{T}(\theta,\mathcal L|_{\mathcal S})=\{\,(\tilde\theta,\varphi)\ |\
\tilde\theta\in\mathcal S,\ \mathrm{T}(\theta,\tilde\theta)\not=\varnothing,\
\varphi\in\mathrm{T}(\theta,\tilde\theta)\, \}$
is called the {\em set of admissible transformations of the equation~$\mathcal L^\theta$
in the class~$\mathcal L|_{\mathcal S}$}.
The~set
$\mathrm{T}(\mathcal L|_{\mathcal S})=\{\,(\theta,\tilde\theta,\varphi)\ |\
\theta,\tilde\theta\in\mathcal S,\ \mathrm{T}(\theta,\tilde\theta)\not=\varnothing,\
\varphi\in\mathrm{T}(\theta,\tilde\theta)\, \}$
is called the {\em set of admissible transformations in~$\mathcal L|_{\mathcal S}$}.

\begin{note}
The first set of admissible transformations was described by Kingston and Sophocleous
in~\cite{Kingston&Sophocleous1991} for a class of generalized Burgers equations.
These authors call transformations of such type as
{\em form-preserving ones}~\cite{Kingston&Sophocleous1998,Kingston&Sophocleous2001}.
The admissible transformations of a class of variable-coefficient reaction--diffusion equations
were studied in~\cite{Vaneeva&Johnpillai&Popovych&Sophocleous2006}.
The admissible transformations of different classes of nonlinear Schr\"odinger equations
are exhaustively described in the series of papers
\cite{Ivanova&Popovych&Eshraghi2005GammaNormalizedSerbia,Popovych&Eshraghi2004Norm,Popovych&Ivanova&Eshraghi2004Cubic,
Popovych&Ivanova&Eshraghi2004Gamma,Popovych&Kunzinger&Eshraghi2006}
in terms of normalized classes of differential equations.
Admissible transformations within the infinitesimal approach were studied in~\cite{Borovskikh2004}.
\end{note}

\begin{note}
In the case of one dependent variable ($m=1$) all the above and below notions can be extended to
contact transformations.
\end{note}

The notions introduced in Section~\ref{SectionOnGrClMethod} can be re-defined in terms of admissible transformations.
Thus, e.g., the point symmetry group~$G^\theta$ of the system~$\mathcal L^\theta$
coincides with~$\mathrm{T}(\theta,\theta)$.
Any element~$\Phi$ from the \emph{usual equivalence group}~$G^{\Equiv}(\mathcal L|_{\mathcal S})$
is a point transformation in the space of $(x,u_{(\prho)},\theta)$,
which is projectible on the space of $(x,u_{(\prho')})$ for any $0\le\prho'\le\prho$,
and $\Phi|_{(x,u_{(\prho')})}$ being the $\prho'$-th order prolongation of $\Phi|_{(x,u)}$,
and $\forall\theta\in\mathcal S$: $\Phi\theta\in\mathcal S$
and $\Phi|_{(x,u)}\in\rm{T}(\theta,\Phi\theta)$.

Similarly, any element~$\Phi$ from the \emph{generalized equivalence group}~$G^{\sim}_{\rm gen}$
is a point transformation in $(x,u,\theta)$-space
such that $\forall\theta\in\mathcal S$: $\Phi\theta\in\mathcal S$
and $\Phi(\cdot,\cdot,\theta(\cdot,\cdot))|_{(x,u)}\in\mathrm{T}(\theta,\Phi\theta)$.
Roughly speaking, $G^{\Equiv}(\mathcal L|_{\mathcal S})$ is the set of admissible transformations
which can be applied to any~$\theta\in\mathcal S$
and~$G^{\Equiv}_{\rm gen}(\mathcal L|_{\mathcal S})$ is formed by the admissible transformations
which can be separated to classes parameterized with $\theta$ running the whole~$\mathcal S$.

The \emph{extended equivalence group} $\bar G^{\sim}=\bar G^{\sim}(\mathcal L|_{\mathcal S})$
of the class~$\mathcal L|_{\mathcal S}$
is formed by the transformations each of which is represented by the pair $\Phi=(\check\Phi,\hat\Phi)$.
Here $\check\Phi$ is a one-to-one mapping in the set of arbitrary elements assumed
as functions of $(x,u_{(p)})$ and
$\hat\Phi=\Phi|_{(x,u)}$ is a point transformation of $(x,u)$ belonging to $\mathrm T(\theta,\check\Phi\theta)$
for any $\theta$ from $\mathcal S$.

The \emph{generalized extended equivalence group} $\bar G^{\sim}_{\rm gen}=\bar G^{\sim}_{\rm gen}(\mathcal L|_{\mathcal S})$
of the class~$\mathcal L|_{\mathcal S}$ consists of transformations each of which is represented by the tuple
$\smash{\Phi=(\check\Phi,\{\hat\Phi^\theta,\theta\!\in\!\mathcal S\})}$.
Here $\smash{\check\Phi}$ is a one-to-one mapping in the set of arbitrary elements assumed
as functions of $(x,u_{(p)})$ and
for any $\theta$ from $\mathcal S$ the element $\smash{\hat\Phi^\theta=\Phi|_{(x,u)}^\theta}$
is a point transformation of $(x,u)$ belonging to $\mathrm T(\theta,\check\Phi\theta)$.

Imposing additional constraints on arbitrary elements, we may single out a subclass in the class under consideration
whose equivalence group is not contained in the equivalence group of the whole class.
Let $\mathcal L|_{\mathcal S'\!}$ be the subclass of the class
$\mathcal L|_{\mathcal S}$, which is constrained by the additional
system of equations $S'(x,u_{( p)},\theta_{(q')})=0$ and inequalities $\Sigma'(x,u_{( p)},\theta_{(q')})\ne0$
with respect to the arbitrary elements $\theta=\theta(x,u_{( p)})$.
($\Sigma'$ can be the 0-tuple.)
Here $\mathcal S'\subset\mathcal S$ is the set of solutions of the united system $S=0$, $S'=0$, $\Sigma\,\Sigma'\ne0$.
We assume that the united system is compatible for the subclass $\mathcal L|_{\mathcal S'\!}$ to be nonempty.
The equivalence group $G^{\sim}(\mathcal L|_{\mathcal S'\!})$ of the subclass $\mathcal L|_{\mathcal S'\!}$
is called a \emph{conditional equivalence group} of the whole class $\mathcal L|_{\mathcal S}$ under the conditions $S'=0$, $\Sigma'\ne0$.
The conditional equivalence group $G^{\sim}(\mathcal L|_{\mathcal S'\!})$ of the class $\mathcal L|_{\mathcal S}$
under the additional conditions $S'=0$, $\Sigma'\ne0$ is called \emph{maximal} if for any subclass $\mathcal L|_{\mathcal S''\!}$
of the class $\mathcal L|_{\mathcal S}$ containing the subclass $\mathcal L|_{\mathcal S'\!}$
we have $G^{\sim}(\mathcal L|_{\mathcal S'\!})\nsubseteq G^{\sim}(\mathcal L|_{\mathcal S''\!})$.
Only maximal conditional equivalence groups are interesting.

The equivalence group $G^{\sim}(\mathcal L|_{\mathcal S})$ generates an equivalence relation on the set of
pairs of additional auxiliary conditions and the corresponding conditional equivalence groups.
Namely, if a transformation from $G^{\sim}(\mathcal L|_{\mathcal S})$ transforms the system
$S'=0$, $\Sigma'\ne0$ to the system $S''=0$, $\Sigma''\ne0$
then the conditional equivalence groups
$G^{\sim}(\mathcal L|_{\mathcal S'\!})$ and $G^{\sim}(\mathcal L|_{\mathcal S''\!})$
are similar with respect to this transformation and will be called \emph{$G^{\sim}$-equivalent}.
If a conditional equivalence group is maximal then any conditional equivalence group $G^{\sim}$-equivalent to it
is also maximal.

Building on the concept of conditional equivalence,
we can formulate the problem of describing $\mathrm{T}(\mathcal L|_{\mathcal S})$
analogously to the usual group classification problem.
Nontrivial additional auxiliary conditions for arbitrary elements naturally arise when
studying~$\mathrm{T}(\mathcal L|_{\mathcal S})$.
Typically, the following steps have to be carried out:
\begin{enumerate}\itemsep=0ex
\item
Construction of~$G^{\sim}(\mathcal L|_{\mathcal S})$ (or $G^{\sim}_{\rm gen}(\mathcal L|_{\mathcal S})$, etc.).
\item
Classification of conditional equivalence groups in the class~$\mathcal L|_{\mathcal S}$, i.e.,
searching for a complete family of $G^{\sim}$-inequivalent additional auxiliary conditions
$S_\gamma=0$,  $\Sigma_\gamma\ne0$, $\gamma\in\Gamma$, such that
$G^{\sim}(\mathcal L|_{\mathcal S_\gamma})$ is a maximal conditional equivalence group of the class $\mathcal L|_{\mathcal S}$
for any $\gamma\in\Gamma$.
Here $\mathcal S_\gamma\subset\mathcal S$ is the set of solutions of
the united system $S=0$, $S_\gamma=0$, $\Sigma\,\Sigma_\gamma\ne0$.

\item
Classification of admissible transformations in the class~$\mathcal L|_{\mathcal S}$,
which do not belong to any conditional equivalence groups (purely {\em partial} equivalence transformations),
i.e.,
searching for a complete family of $G^{\sim}$-inequivalent pairs $(\tilde{\mathcal S}_\lambda,\Phi_\lambda)$, $\lambda\in\Lambda$,
where for every $\lambda\in\Lambda$
the subset $\tilde{\mathcal S}_\lambda$ of~$\mathcal S$ is associated with a well-defined subclass
$\mathcal L|_{\tilde{\mathcal S}_\lambda}$ of $\mathcal L|_{\mathcal S}$,
the point transformation~$\Phi_\lambda$ does not belong to the equivalence group of any subclass of $\mathcal L|_{\mathcal S}$
containing the subclass $\mathcal L|_{\tilde{\mathcal S}_\lambda}$
and
$\Phi_\lambda(\mathcal L|_{\tilde{\mathcal S}_\lambda})\subset\mathcal L|_{\mathcal S}$.
\end{enumerate}

It is obvious that any conditional equivalence is a partial one under the same additional constraint
and any point symmetry transformation for a fixed value $\theta=\theta^0(x,u_{(p)})$
is a partial equivalence transformation under the constraint $\theta=\theta^0.$

Actually, the proposed procedure is far from optimal.
More elaborate techniques are based on the notion of normalized classes
\cite{Popovych&Kunzinger&Eshraghi2006}.

\section{Equivalence transformations and choice \\of investigated class}\label{SectionOnEquivTransOfEqDKfgh}

We start the investigation of Lie symmetry properties of equations from class~\eqref{eqDKfgh}
by finding equivalence groups of this class.

The usual equivalence group~$G^{\sim}$ of class~\eqref{eqDKfgh} is formed by the nondegenerate point transformations
in the space of~$(t,x,u,f,g,h,A,B)$, which are projectible on the space of~$(t,x,u)$,
i.e., they have the form
\begin{gather}
(\tilde t,\tilde x,\tilde u)=(T^t,T^x,T^u)(t,x,u), \nonumber\\[0.5ex]
(\tilde f,\tilde g,\tilde h,\tilde A,\tilde B)=(T^f,T^g,T^h,T^A,T^B)(t,x,u,f,g,h,A,B)\label{EquivTransformations}
\end{gather}
and transform any equation from class~\eqref{eqDKfgh} for the function $u=u(t,x)$
with the arbitrary elements $(f,g,h,A,B)$
to an equation from the same class for the function $\tilde u=\tilde u(\tilde t,\tilde x)$
with the new arbitrary elements~$(\tilde f,\tilde g,\tilde h,\tilde A,\tilde B)$.

\begin{theorem}\label{TheoremOnUsualEquivGroupOfEqDKfgh}
The group $G^{\sim}$ consists of the transformations
\begin{gather*}
\tilde t=\delta_1 t+\delta_2,\quad
\tilde x=X(x), \quad
\tilde u=\delta_3 u+\delta_4, \\
\tilde f=\dfrac{\varepsilon_1\delta_1}{X_x} f, \quad
\tilde g=\varepsilon_1\varepsilon_2^{-1}X_x\, g, \quad
\tilde h=\varepsilon_1\varepsilon_3^{-1}h, \quad
\tilde A=\varepsilon_2A, \quad
\tilde B=\varepsilon_3B,
\end{gather*}
where $\delta_j$ $(j=1,\dots,4)$ and $\varepsilon_i$ $(i=1,\dots,3)$ are arbitrary constants,
$\delta_1\delta_3\varepsilon_1\varepsilon_2\varepsilon_3\not=0$, $X$ is an arbitrary smooth function of~$x$, $X_x\not=0$.
\end{theorem}

It appears that class~\eqref{eqDKfgh} admits other equivalence transformations which do not belong to~$G^{\sim}$
and form, together with usual equivalence transformations, an {\it extended equivalence group}.
We demand for these transformations to be point with respect to $(t,x,u)$.
The explicit form of the new arbitrary elements~$(\tilde f,\tilde g,\tilde h,\tilde A,\tilde B)$ is determined
via $(t,x,u,f,g,h,A,B)$ in some non-fixed (possibly, nonlocal) way.
We construct the complete (in this sense) extended equivalence group~$\hat G^{\sim}$ of
class~\eqref{eqDKfgh}, using the direct method.

\begin{theorem}
The equivalence group $\hat G^{\sim}$ is formed by the transformations
\begin{gather*}
\tilde t=\delta_1 t+\delta_2,\quad
\tilde x=X(x), \quad
\tilde u=\delta_3 u+\delta_4, \\
\tilde f=\dfrac{\varepsilon_1\delta_1\varphi}{X_x}f, \quad
\tilde g=\varepsilon_1\varepsilon_2^{-1}X_x\varphi\,g, \quad
\tilde h=\varepsilon_1\varepsilon_3^{-1}\varphi\,h, \quad
\tilde A=\varepsilon_2 A, \quad
\tilde B=\varepsilon_3 (B+\varepsilon_4 A),
\end{gather*}
where $\delta_j$ $(j=1,\dots,4)$ and $\varepsilon_i$ $(i=1,\dots,4)$ are arbitrary constants,
$\delta_1\delta_3\varepsilon_1\varepsilon_2\varepsilon_3\not=0$,
$X$ is an arbitrary smooth function of~$x$, $X_x\not=0$,
$\varphi=e^{-\varepsilon_4\int \frac{h(x)}{g(x)}dx}$.
\end{theorem}

\noprint{
\checkit[$d\tilde x=X_xdx$, $v=\int \frac{h(x)}{g(x)}dx$,
$\tilde v=\int \frac{\tilde h}{\tilde g}d\tilde x=\frac{\varepsilon_2}{\varepsilon_3}\int \frac hg\frac{d\tilde x}{X_x}$]{}
}

\looseness=-1
Existence of such equivalence transformations can be explained in many respects by features
of representation of equations in the form~\eqref{eqDKfgh}.
This form leads to an ambiguity since the same equation has an infinite series of
different representations. More exactly, two representations~\eqref{eqDKfgh}
with the arbitrary element tuples $(f,g,h,A,B)$ and $(\tilde f,\tilde g,\tilde h,\tilde A,\tilde B)$
determine the same equation if and only if \begin{gather}
\tilde f=\varepsilon_1\varphi\, f, \quad
\tilde g=\varepsilon_1\varepsilon_2^{-1}\varphi\, g, \quad
\tilde h=\varepsilon_1\varepsilon_3^{-1}\varphi\, h,\quad
\tilde A=\varepsilon_2 A, \quad
\tilde B=\varepsilon_3 (B+\varepsilon_4 A),
\label{GaugeEquivTransformationsDKfgh}
\end{gather}
where $\varphi=e^{-\varepsilon_4\int \frac{h(x)}{g(x)}dx}$,
$\varepsilon_i$ $(i=1,\dots,4)$ are arbitrary constants, $\varepsilon_1\varepsilon_2\varepsilon_3\not=0$
(the variables $t$, $x$ and $u$ are not transformed!).

Transformations~\eqref{GaugeEquivTransformationsDKfgh} act only on arbitrary elements, but not on the variables $t,x,u$
and therefore, do not really change the equations.
In general, transformations of such type can be considered as trivial~\cite{LisleDissertation}
(``gauge'') equivalence transformations
and form the ``gauge'' (normal) subgroup~$\hat G^{\sim g}$ of the extended equivalence group~$\hat G^{\sim}$.

We note that transformations~\eqref{GaugeEquivTransformationsDKfgh} with $\varepsilon_4\not=0$ are
nonlocal with respect to arbitrary elements, otherwise they belong to~$G^{\sim}$
and form the ``gauge'' (normal) subgroup~$G^{\sim g}$ of the equivalence group~$G^{\sim}$.

The factor-group $\hat G^{\sim}/\hat G^{\sim g}$ coincides for class~\eqref{eqDKfgh} with~$G^{\sim}/G^{\sim g}$
and can be assumed to consist of the transformations
\begin{equation} \label{RealEquivTransformationsDKfgh}\arraycolsep=0em
\begin{array}{l}
\tilde t=\delta_1 t+\delta_2,\quad
\tilde x=X(x), \quad
\tilde u=\delta_3 u+\delta_4,\\[1ex]
\tilde f=\dfrac{\delta_1}{X_x} f, \quad
\tilde g=X_x\,g, \quad
\tilde h=h, \quad
\tilde A=A, \quad
\tilde B=B,
\end{array}
\end{equation}
where $\delta_i$ ($i=1,\dots,4$) are arbitrary constants, $\delta_1\delta_3\not=0$,
$X$ is an arbitrary smooth function of~$x$, $X_x\not=0$.

\begin{note}
After extending the set of arbitrary elements of class~\eqref{eqDKfgh} with one more arbitrary element $l=l(x)$
constrained by the equation $l_x=h/g$, we can consider
the extended equivalence group $\hat G^{\sim}$ of class~\eqref{eqDKfgh}, up to gauge transformations of translations with respect to~$l$,
as the usual equivalence group of the modified class.
The transformation components corresponding to the arbitrary elements depending on~$x$ take the form
\[
\tilde f=\dfrac{\varepsilon_1\delta_1e^{-\varepsilon_4l}}{X_x}f, \quad
\tilde g=\varepsilon_1\varepsilon_2^{-1}X_xe^{-\varepsilon_4l}\,g, \quad
\tilde h=\varepsilon_1\varepsilon_3^{-1}e^{-\varepsilon_4l}\,h, \quad
\tilde l=\varepsilon_2\varepsilon_3^{-1}l+\varepsilon_5.
\]
\end{note}

It was convenient during investigation of the subclass of the class~(\ref{eqDKfgh}) with
the additional condition $h=1$ to gauge the parameter-function~$g$ to 1~\cite{Popovych&Ivanova2004NVCDCEs}.
In the same way, using the transformation $\tilde t=t$, $\tilde x=\int \frac{dx}{g(x)}$, $\tilde u=u$
from $G^{\sim}/G^{\sim g}$, we can reduce equation~(\ref{eqDKfgh}) to
\[
\tilde f(\tilde x)\tilde u_{\tilde t}= (A(\tilde u)
\tilde u_{\tilde x})_{\tilde x} + \tilde h(\tilde x)B(\tilde u)\tilde u_{\tilde x},
\]
where $\tilde f(\tilde x)=g(x)f(x)$, $\tilde g(\tilde x)=1$ and $\tilde h(\tilde x)=h(x)$.
(Likewise any equation of form~\eqref{eqDKfgh} can be reduced to the same form with $\tilde f(\tilde x)=1.$)
That is why, without loss of generality we can restrict ourselves to investigation of the equation
\begin{equation} \label{eqDKfh}
f(x)u_t=\left(A(u)u_x \right)_x + h(x)B(u)u_x.
\end{equation}

Any transformation from~$\hat G^{\sim}$, which preserves the condition $g=1$, has the form
\[
\arraycolsep=0em
\begin{array}{l}
\tilde t=\delta_1 t+\delta_2,\quad
\tilde x=\delta_5 \int e^{\delta_8\int\! h}dx+\delta_6, \quad
\tilde u=\delta_3 u+\delta_4,\\[1ex]
\tilde f=\delta_1\delta_5^{-1}\delta_9 fe^{-2\delta_8\int\! h}, \quad
\tilde h=\delta_9\delta_7^{-1} he^{-\delta_8\int\! h}, \quad
\tilde g=g, \\[1ex]
\tilde A=\delta_5\delta_9A, \quad
\tilde B=\delta_7(B+\delta_8A),
\end{array}
\]
where $\delta_i$ ($i=1,\dots,9$) are arbitrary constants, $\delta_1\delta_3\delta_5\delta_7\delta_9\not=0$
and $\int\! h=\int\! h(x)\,dx$.
The set of such transformations is a subgroup of~$\hat G^{\sim}$.
Its projection~$\hat G^{\sim}_1$ to the condition $g=1$ can be considered as the generalized extended equivalence group of class~\eqref{eqDKfh} after
admitting dependence of transformations of variables on arbitrary elements~\cite{Meleshko1994}
and additional supposition that such dependence can be nonlocal~\cite{Popovych&Kunzinger&Eshraghi2006}.
The group~$G^{\sim}_1$ of usual (local) equivalence transformations of class~\eqref{eqDKfh}
coincides with the subgroup singled out from~$\hat G^{\sim}_1$ via the condition $\delta_8 = 0$.
The transformations from~$\hat G^{\sim}_1$ with non-vanishing values of the parameter $\delta_8$
are nonlocal in the arbitrary element~$h$ and are projections of compositions of usual equivalence
and nonlocal gauge transformations from~$\hat G^{\sim}$.

There exists a way to avoid operations with nonlocal equivalence transformations.
More exactly, we can assume that the parameter-function $B$ is determined up to an additive
term proportional to $A$ and subtract such term from $B$ before applying equivalence
transformations~\eqref{RealEquivTransformationsDKfgh}.

At the same time, there is another possible generalization of the gauge $g=1$ from
the case $h=1$ to the general case of $h$, namely the gauge $g=h$.
Any equation of the form~\eqref{eqDKfgh} can be reduced to the equation
\begin{equation} \label{eqDKfhh}
f(x)u_t=\left(h(x)A(u)u_x \right)_x + h(x)B(u)u_x
\end{equation}
by the transformation
$\tilde t=t$, $\tilde x=\int \frac{h(x)}{g(x)}dx$, $\tilde u=u$ from $G^{\sim}/G^{\sim g}$.
The usual equivalence group~$G^{\sim}_h$ of the subclass~\eqref{eqDKfhh} consists of the transformations
\[
\arraycolsep=0em
\begin{array}{l}
\tilde t=\delta_1 t+\delta_2,\quad
\tilde x=\delta_5 x+\delta_6, \quad
\tilde u=\delta_3 u+\delta_4,\\[1ex]
\tilde f=\delta_1\delta_5^{-1}\delta_9 fe^{\delta_8x}, \quad
\tilde h=\delta_9\delta_7^{-1} he^{\delta_8x}, \quad
\tilde A=\delta_5A, \quad
\tilde B=\delta_7(B-\delta_8A),
\end{array}
\]
where $\delta_i$ ($i=1,\dots,9$) are arbitrary constants, $\delta_1\delta_3\delta_5\delta_7\delta_9\not=0$.
It is the projection of the subgroup of~$\hat G^{\sim}$ preserving the constraint $g=h$ to the subclass~\eqref{eqDKfhh}.
One of the advantages of the gauge~$g=h$ over the gauge~$g=1$ is that
the generalized extended equivalence group of class~\eqref{eqDKfhh} coincides with $G^{\sim}_h$.

To simplify our consideration, we will use the gauges $g=1$ and $g=h$ simultaneously.
Any equation with $g=1$ can be reduced to that with $\tilde g=\tilde h$ by means of
the simple equivalence transformation
\begin{equation}\label{EqTransFromDKf1hToDKfhh}
\tilde t=t, \quad {\textstyle\tilde x=\int h\, dx,} \quad \tilde u=u, \quad \tilde f=\frac{f}{h}, \quad \tilde h=h,
 \quad \tilde A=A, \quad \tilde B=B.
\end{equation}


\section{Group classification of diffusion--convection equations}\label{SectionOnGrClOfDCEfgh}
We consider a one-parameter Lie group of point transformations in $(t,x,u)$
with an infinitesimal operator of the form
$
Q=\tau(t,x,u)\p_t+\xi(t,x,u)\p_x+\eta(t,x,u)\p_u,
$
which leaves equation~\eqref{eqDKfgh} invariant.
The Lie criterion of infinitesimal invariance yields the following system of determining equations for
$\tau,\ \xi$ and $\eta$:
\begin{gather}
\tau_x=\tau_u=\xi_u=0,\quad \eta_{uu}=0,\label{deteq1}\\
\xi\frac{f_x}{f}-\xi\frac{g_x}{g}-\tau_t+2\xi_x=\eta\frac{A_u}{A},\label{deteq2}\\
(g\eta_x)_xA+h\eta_xB=\eta_tf,\label{deteq3}\\
(g_x\eta+2g\eta_x)A_u+\left((2\eta_{xu}-\xi_{xx})g+\left(\tau_t-\xi_x-\xi\frac{f_x}{f}\right)g_x +\xi g_{xx}\right)A
\nonumber\\
{}+h\eta B_u+\left(\tau_t-\xi_x-\xi\frac{f_x}{f}+\xi \frac{h_x}{h} \right)hB+\xi_tf=0.\label{deteq4}
\end{gather}

Equations~\eqref{deteq1} do not contain arbitrary elements. Integration of them yields
\begin{equation}\label{dse_modified}
\tau=\tau(t),\quad \xi=\xi(t,x),\quad \eta=\eta^1(t,x)u+\eta^0(t,x).
\end{equation}

Thus, group classification of~\eqref{eqDKfgh} reduces to solution
of classifying conditions~\eqref{deteq2}--\eqref{deteq4}.

Splitting system~\eqref{deteq2}--\eqref{deteq4}
with respect to the arbitrary elements and their non-vanishing derivatives gives
the equations $\tau_t=0,$ $\xi=0,$ $\eta=0$
on the coefficients of operators from the Lie algebra $A^\cap$ of the kernel of principal groups of~\eqref{eqDKfgh}.
As a result, the following theorem is true.

\begin{theorem}\label{th1}
The Lie algebra of the kernel of principal groups of~\eqref{eqDKfgh} is $A^\cap=\langle \p_t \rangle$.
\end{theorem}

Studying all possible cases of integration of equations~\eqref{deteq2}--\eqref{deteq4}
under condition~\eqref{dse_modified} up to the extended equivalence group~$\hat G^{\Equiv}$
in the both gauges~$g=1$ and $g=h$ leads to the following theorem.

\begin{theorem}\label{TheoremOnClassificationOfeqDKfgh}
A complete set of $\hat G^{\Equiv}$-inequivalent equations~\eqref{eqDKfgh} which have
the wider Lie invariance algebras than $A^\cap$ is exhausted by cases given in tables~1--3
or tables~1$'$--3$'$.
\end{theorem}

The proof of theorem~\ref{TheoremOnClassificationOfeqDKfgh} is briefly sketched in the next section.

In tables~1--3 and in tables~1$'$--3$'$ we list all possible $\hat G^{\Equiv}$-inequivalent
sets of functions $f(x)$, $h(x)$, $A(u)$ and $B(u)$ with different Lie symmetry properties and the corresponding invariance algebras
under the gauges~$g=1$ and $g=h$, respectively.
Detailed explanatory notes to classification results are presented only for the first gauge.
We give the same numbers for the corresponding ($\hat G^{\Equiv}$-equivalent) cases in the gauges $g=1$ and $g=h$.
The asterisked cases from tables~2 and~3 are equivalent to the cases from tables~2$'$ and~3$'$ with
the same numbers, where the parameter-function~$h$ takes the value $h=x$.
The similar non-asterisked cases correspond to the same cases from tables~2$'$ and~3$'$,
where
\[
p'=\dfrac{p-q}{q+1}, \quad q'=\dfrac{q}{q+1} \qquad \mbox{or} \qquad
p'=-\dfrac{1}{p+2} \quad \mbox{if} \quad q=p+1.
\]
For convenience we use double numeration T.N of classification cases and point equivalence
transformations, where T denotes the number of table and N does the number of case (or transformation) in table~T.
The notion ``equation~T.N'' is used for the equation of form~(\ref{eqDKfh}) where
the parameter-functions take values from the corresponding case.

The operators from tables~1--3 or tables~1$'$--3$'$ form bases of the maximal invariance algebras
if and only if the associated sets of the arbitrary elements $f$, $h$, $A$ and $B$ are
$\hat G^{\Equiv}$-inequivalent to ones with more extensive invariance algebras.
For example, the operators from case~3.1 have the above property if and only if $f\not=f^3$,
where the expression for~$f^3$ is given after table~3.
We indicate only some of constraints on parameters which arise in such way.

In spite of the fact that for the classification we used the generalized extended equivalence group,
the classification lists include similar equations that are equivalent only with respect to additional equivalence transformations.
Class~\eqref{eqDKfgh} is very reach from this point of view.
As one can see from the tables, there exist a lot of nontrivial additional transformations between the classified equations.
In fact, for both gauges $g=1$ and $g=h$ we carried out two essentially different classification:
the classification with respect to the generalized equivalence group~$\hat G^{\Equiv}$ and the classification with respect to
the set of all possible point transformations.

Numbers with the same Arabic numerals and different Roman letters
correspond to cases that are equivalent with respect to additional equivalence transformations.
Explicit formulas for these transformations are presented after the corresponding tables~1--3.
Any additional equivalence transformation between cases from tables~1$'$--3$'$ is obtained via the composition of
the inverse of transformation~\eqref{EqTransFromDKf1hToDKfhh},
an additional equivalence transformation between the corresponding cases from tables~1--3
and transformation~\eqref{EqTransFromDKf1hToDKfhh}.
The cases which are contained in tables with different Arabic numbers
or are numbered with different Arabic numerals are reciprocally inequivalent with respect to point transformations.
The exclusion is case 3.\ref{A1Bffhh} (resp.\ 3$'$.\ref{a1bafexhghhgex})
which is reduced to a subcase of case 1.\ref{gcAaBafexh1}a (resp.\ 1$'$.\ref{AaBafepxghh1}a).

\begin{theorem}\label{TheoremOnClassificationOfDCEfghWRTPointTrans}
Up to point transformations, a complete list of extensions of the maximal Lie invariance group of equations from class~\eqref{eqDKfgh}
is exhausted by the cases from tables~1--3 (resp.\ tables~1$'$--3$'$) numbered with Arabic numbers without Roman letters and
subcases ``a'' of each multi-case, excluding case 3.\ref{A1Bffhh} (resp.\ 3$'$.\ref{a1bafexhghhgex}).
\end{theorem}

The proof of theorem~\ref{TheoremOnClassificationOfDCEfghWRTPointTrans} involves arguments on differences
in structure of maximal Lie invariance algebras of listed cases and
preliminary description of admissible transformations given in Section~\ref{SectionAdmisTransf}.
We plan to present it in the new version of the second part of the series.
The multifarious structure of additional equivalence transformations of class~\eqref{eqDKfgh}
displays a  structure complexity of the entire set of admissible transformations.

Analyzing the classification results in a way similar to~\cite{Popovych&Ivanova2004NVCDCEs}
leads to the following theorem.

\begin{theorem}\label{TheoremOnReducibilityOfEqsWith4DimAlgs}
If an equation of form~\eqref{eqDKfgh} is invariant with respect to a Lie algebra
of dimension not less than 4 then it can be reduced by a point transformation to a one with $f=g=h=1$.
\end{theorem}

\begin{note}
The simultaneous usage of two gauges $g=1$ and $g=h$ also allows us to explain singularity
of some values of parameters with respect to Lie symmetry properties.
For instance, in case~3.\ref{AumB0f1h}e the singular values of $\mu$ are $-2$, $-4/3$ and~$-1$.
Singularity of $\mu=-1$ is obvious since $\mu+1$ is in the denominator of the power of $|x|$ in~$f$.
Singularity of $\mu=-4/3$ can be explained by the fact that the power in~$f$ vanishes for this
value of~$\mu$.
At the same time, singularity of $\mu=-2$ becomes apparent only after change of the gauge $g=1$
to $g=h$ where the power of $|x|$ in~$f$ equals $(\mu+1)/(\mu+2)$.
\end{note}

\newpage

\setcounter{tbn}{0}

\begin{center}\footnotesize\renewcommand{\arraystretch}{1.1}
Table~1. Case of $\forall A(u)$ (gauge $g=1$)\\[1ex]
\begin{tabular}{|l|c|c|c|l|}
\hline
N & $B(u)$ & $f(x)$ & $h(x)$ & \hfil Basis of A$^{\rm max}$ \\
\hline
\refstepcounter{tbn}\label{gcAaBafaha}\thetbn & $\forall$ & $\forall$ & $\forall$ & $\p_t$ \\
\refstepcounter{tbn}\label{gcAaBafexh1}
\thetbn a&$\forall$ & $e^{p x}$ & 1 &$\p_t,\, p t\p_t+\p_x$ \\
\thetbn a$'$ & $\forall$ & $|x|^p$ & $x^{-1}$ & $\p_t,\, (p+2)t\p_t+x\p_x$ \\
\thetbn b & 1 & $e^x$ & $e^x+\beta$ & $\p_t,\, e^{-t}(\p_t-\p_x)$ \\
\thetbn c & 1 & $|x|^p$ & $x|x|^p+\beta x^{-1}$ & $\p_t,\, e^{-(p+2)t}(\p_t-x\p_x)$ \\
\refstepcounter{tbn}\label{gcAaB1fx-2hlnxx}
\thetbn & 1 & $x^{-2}$ & $x^{-1}\ln|x|$ & $\p_t,\, e^{-t}x\p_x$ \\
\refstepcounter{tbn}\label{gcAaB0f1ha}
\thetbn a& 0 & 1 & 1 & $\p_t,\, \p_x,\, 2t\p_t+x\p_x$ \\
\thetbn b& 1 & 1 & $x$ & $\p_t,\, e^{-t}\p_x,\, e^{-2t}(\p_t-x\p_x)$ \\
\thetbn c& 1 & 1 & 1 & $\p_t,\, \p_x,\, 2t\p_t+(x-t)\p_x$ \\
\hline
\end{tabular}
\end{center}
{\footnotesize
\noindent
Here $p\in\{0,1\}$ in case~\ref{gcAaBafexh1}a;
$p\ne-2$ in case \ref{gcAaBafexh1}c;
$\beta\in\{0,\pm 1\}$ in cases~\ref{gcAaBafexh1}b and~\ref{gcAaBafexh1}c.
Case \ref{gcAaBafexh1}a$'$ is equivalent to case \ref{gcAaBafexh1}a with respect to the transformation
$
\tilde t=t,\ \tilde x=\ln|x|,\ \tilde u=u,\ \tilde A=A,\ \tilde B=B-A,\ \tilde p=p+2
$
from $\hat G^\sim_1$. It is given for the convenience of presentation of results only.

\smallskip

\setcounter{casetran}{0}\noindent
Additional equivalence transformations:\\[1ex]
\refstepcounter{casetran}\thecasetran. \ref{gcAaBafexh1}a($p=0$, $B=1$)
$\to$ \ref{gcAaBafexh1}a($p=0$, $B=0$):\quad
$\tilde t=t$, $\tilde x=x+t$, $\tilde u=u$;
\\[.7ex]
\thecasetran$'$. \ref{gcAaBafexh1}a$'$($p=-2$, $B=1$) $\to$
\ref{gcAaBafexh1}a$'$($p=-2$, $B=0$):\quad $\tilde t=t$, $\tilde x=xe^t$, $\tilde u=u$;
\\[.7ex]
\refstepcounter{casetran}\thecasetran. \ref{gcAaBafexh1}b $\to$ \ref{gcAaBafexh1}a($B=\beta$, $p=1$):\quad
$\tilde t=e^t$, $\tilde x=x+t$, $\tilde u=u$;
\\[.7ex]
\refstepcounter{casetran}\thecasetran. \ref{gcAaBafexh1}c($p\ne-2$) $\to$ \ref{gcAaBafexh1}a$'$($p\ne-2$):\quad
$\tilde t=(e^{(p+2)t}-1)/(p+2)$, $\tilde x=xe^t$, $\tilde u=u$;
\\[.7ex]
\refstepcounter{casetran}\thecasetran. \ref{gcAaB0f1ha}b $\to$ \ref{gcAaB0f1ha}a:\quad
$\tilde t=e^{2t}/2$, $\tilde x=xe^t$, $\tilde u=u$;
\\[.7ex]
\refstepcounter{casetran}\thecasetran. \ref{gcAaB0f1ha}c $\to$ \ref{gcAaB0f1ha}a:\quad
$\tilde t=t$, $\tilde x=x+t$, $\tilde u=u$.
\par}

\setcounter{tbn}{0}

\begin{center}\footnotesize\renewcommand{\arraystretch}{1.1}
Table~2. Case of $A(u)=e^{\mu u}$ (gauge $g=1$)\\[1ex]
\begin{tabular}{|l|c|c|c|l|}
\hline
N & $B(u)$ & $f(x)$ & $h(x)$ & \hfil Basis of A$^{\rm max}$ \\
\hline
\refstepcounter{tbn}\label{AeuB0faha}\thetbn & 0 & $\forall$ & 1 & $\p_t,\, t\p_t-\p_u$ \\
\refstepcounter{tbn}\label{AeuBenufemxhex}\thetbn & $e^{\nu u}$ & $|x|^p$ & $|x|^q$ &
$\p_t,\, (p\mu-p\nu-2\nu-q\mu+\mu)t\p_t+(\mu-\nu)x\p_x+(q+1)\p_u$ \\
\thetbn${}^*$ & $e^{\nu u}$ & $e^{p x}$ & $\varepsilon e^{x}$ &
$\p_t,\, (p\mu-p\nu-\mu)t\p_t+(\mu-\nu)\p_x+\p_u$ \\
\refstepcounter{tbn}\label{AeuBueufh}\thetbn & $ue^u$ & $h^2e^{q\int\!h}$ & $(h^{-1})''=-2p h$
& $\p_t,\, (2p+q) t\p_t+h^{-1}\p_x-2p\p_u$ \\
\refstepcounter{tbn}\label{AeuBenukf1h1}\thetbn & $e^u+\varkappa$ & 1 & 1 &
$\p_t,\, \p_x,\, (\mu-2)t\p_t+((\mu-1)x+\varkappa t)\p_x+\p_u$ \\
\refstepcounter{tbn}\label{AeuBuf1h1}\thetbn & $u$ & 1 & 1 &$\p_t,\, \p_x,\, t\p_t+(x-t)\p_x+\p_u$ \\
\refstepcounter{tbn}\label{AeuB0ff1h}\thetbn a &
0 & $f^1(x)$ & 1 & $\p_t,\, t\p_t-\p_u,\, (\beta x^2+\gamma_1x+\gamma_0)\p_x+(\beta x+\alpha)\p_u$ \\
\thetbn b &
1 & $|x|^p$ & $\varepsilon x|x|^{p}$ & $\p_t,\, x\p_x+(p+2)\p_u,\, e^{-\varepsilon(p+2)t}(\p_t-\varepsilon x\p_x)$ \\
\thetbn b${}^*$ &
1 & $e^{x}$ & $\varepsilon e^{x}$ & $\p_t,\, \p_x+\p_u,\, e^{-\varepsilon t}(\p_t-\varepsilon\p_x)$ \\
\thetbn c & 1 & $x^{-2}$ & $\varepsilon x^{-1}$ & $\p_t,\, x\p_x,\, t\p_t-\varepsilon tx\p_x-\p_u$ \\
\refstepcounter{tbn}\label{AeuB0f1h}\thetbn a & 0 & 1& 1 & $\p_t,\, t\p_t-\p_u,\, 2t\p_t+x\p_x,\, \p_x$ \\
\thetbn b & 1 & 1 & 1 & $\p_t, \,\p_x,\, t\p_t-t\p_x-\p_u,\, 2t\p_t+(x-t)\p_x$ \\
\thetbn c & 1 & 1 & $\varepsilon x$ & $\p_t,\, x\p_x+2\p_u,\, e^{-\varepsilon t}\p_x,\,
e^{-2\varepsilon t}(\p_t-\varepsilon x\p_x)$ \\
\thetbn d & 0 & $x^{-3}$ & 1 & $ \p_t,\, t\p_t-\p_u,\, x\p_x-\p_u,\, x^2\p_x+x\p_u$ \\
\thetbn e & 1 & $x^{-3}$ & $x^{-2}$ & $\p_t,\, x\p_x-\p_u,\, e^t(\p_t-x\p_x),\, e^t(x^2\p_x+x\p_u)$ \\
\hline
\end{tabular}
\end{center}
{\footnotesize\looseness=-1
Here $(\mu,\,\nu)\in\{(0,\,1),\, (1,\,\nu)\}$, $\nu\ne\mu$ in cases~\ref{AeuBenufemxhex} and \ref{AeuBenufemxhex}${}^*$;
$\mu\ne1$ and $\varkappa\in\{-1,0,1\}$ in case~\ref{AeuBenukf1h1};
$\mu=1$ in cases~\ref{AeuB0faha}, \ref{AeuBuf1h1}--\ref{AeuB0f1h}e;
$q\ne-1$ in case~\ref{AeuBenufemxhex} (otherwise it is a subcase of case~1.\ref{gcAaBafexh1}a$'$);
$\varepsilon=\pm1$ in cases \ref{AeuBenufemxhex}${}^*$, \ref{AeuB0ff1h}b, \ref{AeuB0ff1h}b${}^*$, \ref{AeuB0ff1h}c and \ref{AeuB0f1h}c;
$p\not\in\{-3,-2,0\}$ in case~\ref{AeuB0ff1h}b;
$(\beta,\gamma_1,\gamma_0,\alpha)\in\{(1,0,\pm1,\hat\alpha),(1,0,0,1),(0,1,0,\check\alpha),(0,0,1,1)\}$, $\hat\alpha,\check\alpha=\const$, $\hat\alpha\geqslant0$
and
\[
f^1(x)=\exp\left\{\int\frac{-3\beta x-2\gamma_1+\alpha}{\beta x^2+\gamma_1x+\gamma_0}\,dx\right\}.
\]

\setcounter{casetran}{0}\noindent
Additional equivalence transformations:\\[1ex]
\refstepcounter{casetran}\thecasetran. \ref{AeuBenukf1h1}($\varkappa\ne0$) $\to$ \ref{AeuBenukf1h1}($\varkappa=0$):\quad
$\tilde t=t$, $\tilde x=x+\varkappa t$, $\tilde u=u$;
\\[.7ex]
\refstepcounter{casetran}\thecasetran. \ref{AeuB0ff1h}b $\to$ \ref{AeuB0ff1h}a
($\beta=\gamma_0=0$, $\alpha=(p+2)\gamma_1$):\quad $\tilde t=(e^{\varepsilon(p+2)t}-1)/(\varepsilon(p+2))$, $\tilde x=xe^{\varepsilon t}$, $\tilde u=u$;
\\[.7ex]
\refstepcounter{casetran}\thecasetran. \ref{AeuB0ff1h}b${}^*$ $\to$ \ref{AeuB0ff1h}a
($\beta=\gamma_1=0$, $\alpha=\gamma_0$):\quad $\tilde t=e^{\varepsilon t}/\varepsilon$, $\tilde x=x+\varepsilon t$, $\tilde u=u$;
\\[.7ex]
\refstepcounter{casetran}\thecasetran. \ref{AeuB0ff1h}c $\to$ \ref{AeuB0ff1h}a
($\beta=\gamma_0=\alpha=0$):\quad $\tilde t=t$, $\tilde x=xe^{\varepsilon t}$, $\tilde u=u$;
\\[.7ex]
\refstepcounter{casetran}\thecasetran.
\ref{AeuB0f1h}b$\to$\ref{AeuB0f1h}a:\quad $\tilde t= t$, $\tilde x=x+t$, $\tilde u=u$;
\\[.7ex]
\refstepcounter{casetran}\thecasetran. \ref{AeuB0f1h}c$\to$\ref{AeuB0f1h}a:\quad
$\tilde t=e^{2\varepsilon t}/(2\varepsilon)$, $\tilde x=xe^{\varepsilon t}$, $\tilde u=u$;
\\[.7ex]
\refstepcounter{casetran}\thecasetran.
\ref{AeuB0f1h}d$\to$\ref{AeuB0f1h}a:\quad $\tilde t= t\sign x$, $\tilde x=1/x,$ $\tilde u=u-\ln |x|$;
\\[.7ex]
\refstepcounter{casetran}\thecasetran. \ref{AeuB0f1h}e$\to$\ref{AeuB0f1h}a:\quad $\tilde t=(e^{2t}t\sign x)/2$,
$\tilde x=e^{-t}/x,$ $\tilde u=u-t-\ln |x|$.
\par}

\newpage
\setcounter{tbn}{0}

\begin{center}\footnotesize\renewcommand{\arraystretch}{1.35}
Table~3. Case of $A(u)=|u|^\mu$ (gauge $g=1$)\\[1ex]
\begin{tabular}{|l|l|c|c|c|l|}
\hline
\hfil N & $\hfil \mu$ & $B(u)$ & $f(x)$ & $h(x)$ & \hfil Basis of A$^{\rm max}$ \\
\hline
\refstepcounter{tbn}\label{AumB0fh}\thetbn & $\forall$ & 0 & $\forall$ & 1 & $\p_t,\, \mu t\p_t-u\p_u$ \\
\refstepcounter{tbn}\label{AumBunfxlhxg}\thetbn&
$\forall$ & $|u|^\nu$ & $|x|^p$ & $|x|^q$ & $\p_t,\, (\mu+p\mu-q\mu-p\nu-2\nu)t\p_t$ \\
&&&&&${}+(\mu-\nu)x\p_x+(q+1)u\p_u$ \\
\thetbn${}^*$ &
$\forall$ & $|u|^\nu$ & $e^{px}$ & $\varepsilon e^{x}$ & $\p_t,\, (p\mu-p\nu-\mu)t\p_t+(\mu-\nu)\p_x+u\p_u$ \\
\refstepcounter{tbn}\label{AumBumlnuffhh}\thetbn & $\forall$ & $|u|^\mu\ln|u|$ & $h^2e^{q\int\!h}$
& $(h^{-1})''=-2p h$ & $\p_t,\,(2p\mu+q)t\p_t+h^{-1}\p_x-2pu\p_u$ \\
\refstepcounter{tbn}\label{Au-1B1fx-3ebxhx-2ebx}\thetbn & $\forall$ & 1 & $f^2(x)$ & $\varepsilon xf^2(x)$ & $\p_t,$ \\
&&&&& $e^{\varepsilon t}(\p_t-\varepsilon((\mu+1)\beta x^2+x)\p_x-\varepsilon\beta xu\p_u)$ \\
\refstepcounter{tbn}\label{A1Bafh2hh}\thetbn & 0 & $\forall$ & $h^2$ &
$(h^{-1})''=-2p h$ & $\p_t,\, e^{-2p t}h^{-1}\p_x$ \\
\refstepcounter{tbn}\label{A1Bffhh}\thetbn & 0 & $\forall$ & $x^{-2}/\ln|x|$ &
$x^{-1}/\ln|x|$ & $\p_t,\, e^{t}(\p_t+x\ln|x|\p_x)$ \\
\refstepcounter{tbn}\label{A1Buffhh}\thetbn & 0 & $u$ & $h^2e^{\int\!h}$ &
$\left( h^{-1}\right)''=-2p h$ & $\p_t,\, t\p_t+h^{-1} \p_x-2p\p_u$ \\
\refstepcounter{tbn}\label{AumBunkf1h1}\thetbn & $\forall$ & $|u|^\nu+\varkappa$ & 1 & 1 & $\p_t,\, \p_x,$ \\
&&&&& $(\mu-2\nu)t\p_t+((\mu-\nu)x+\nu\varkappa t)\p_x+u\p_u$ \\
\refstepcounter{tbn}\label{AumBlnuf1h1}\thetbn & $\forall$ & $\ln|u|$ & 1 & 1 &
$\p_t,\, \p_x,\, \mu t\p_t+(\mu x-t)\p_x+u\p_u$ \\
\refstepcounter{tbn}\label{A1Bufh2hh}\thetbn & 0 & $u$ & $h^2$ & $\left(h^{-1}\right)''=-2p h$ &
$\p_t,\, {e^{-2p t}}h^{-1}\p_x,\, h^{-1}\p_x-2p\p_u$ \\
\refstepcounter{tbn}\label{A1Blnufh2hh}\thetbn & 0 & $\ln|u|$ & $h^2$ & $\left(h^{-1}\right)''=-2p h$ &
$\p_t,\, {e^{-2p t}}h^{-1}\p_x,\, h^{-1}\p_x-2p u\p_u$ \\
\refstepcounter{tbn}\label{AumB0ff3hh}\thetbn a & $\forall$ & 0 & $f^3(x)$ & 1 &$\p_t,\, \mu t\p_t-u\p_u,$ \\
&&&&& $\alpha t\p_t+((\mu+1)\beta x^2+\gamma_1x+\gamma_0)\p_x+\beta xu\p_u$ \\
%
%
\thetbn b & $\forall$ & 1 & $|x|^p$ & $\varepsilon x|x|^{p}$ &
$\p_t,\, \mu x\p_x+(p+2)u\p_u,\, e^{-\varepsilon(p+2)t}(\p_t-\varepsilon x\p_x)$ \\
%
\thetbn b${}^*\!\!\!$ &
$\ne-1$ & 1 & $e^{x}$ & $\varepsilon e^{x}$ & $\p_t,\, \mu\p_x+u\p_u,\, e^{-\varepsilon t}(\p_t-\varepsilon\p_x)$ \\
\thetbn c & $\ne-2$ & 1& $x^{-2}$ & $\varepsilon x^{-1}$ & $\p_t,\, x\p_x,\, \mu t\p_t-\varepsilon\mu tx\p_x-u\p_u$ \\
\refstepcounter{tbn}\label{A-65B1fx2hx2}\thetbn & $-6/5$ & 1 & $x^2$ & $x^2$ & $\p_t,\, 2t\p_t+2x\p_x-5u\p_u,$ \\
&&&&& $t^2\p_t+(2tx+x^2)\p_x-5(t+x)u\p_u$ \\
\refstepcounter{tbn}\label{AumB0f1h}\thetbn a & $\ne -4/3$ & 0 & 1 & 1 &
$\p_t,\, \mu t\p_t-u\p_u,\, \p_x,\, 2t\p_t+x\p_x$\\
\thetbn b & $\ne -4/3$ & 1 & 1 & 1 & $\p_t,\, \mu t\p_t-\mu t\p_x-u\p_u,\, \p_x,\, 2t\p_t+(x-t)\p_x$ \\
\thetbn c & $\ne -4/3$ & 1 & 1 & $\varepsilon x$ &
$\p_t,\, \mu x\p_x+2u\p_u,\, e^{-\varepsilon t}\p_x,\, e^{-2\varepsilon t}(\p_t-\varepsilon x\p_x)$ \\
\thetbn d & $\ne -4/3, -1$ & 0 & $|x|^{-\frac{3\mu+4}{\mu+1}}$ & 1 &
$\p_t,\, \mu t\p_t-u\p_u,\, (\mu+2)t\p_t-(\mu+1)x\p_x,$ \\ &&&&& $(\mu+1)x^2\p_x+xu\p_u $ \\
\thetbn e&$\ne -2,$  & 1 & $|x|^{-\frac{3\mu+4}{\mu+1}}$ & $\varepsilon x|x|^{-\frac{3\mu+4}{\mu+1}}$ &
$\p_t,\, \mu(\mu+1)x\p_x-(\mu+2)u\p_u,$ \\
&$-4/3,-1$&&&& $e^{\varepsilon\frac{\mu+2}{\mu+1}t}(\p_t-\varepsilon x\p_x),\,e^{\varepsilon t}((\mu+1)x^2\p_x+xu\p_u)$\\
\thetbn f & $-1$ & 0 & $e^x$ & 1 & $ \p_t,\, t\p_t+u\p_u,\, \p_x-u\p_u,\,2t\p_t+x\p_x-xu\p_u$ \\
\thetbn g & $-1$ & 1 & $e^{x}$ & $\varepsilon e^{x}$ &
$\p_t,\, \p_x-u\p_u,\, (x+\varepsilon t-2)\p_x-(x+\varepsilon t)u\p_u,$ \\
&&&&& $e^{-\varepsilon t}(\p_t-\varepsilon\p_x)$ \\
\thetbn h & $-2$ & 1 & $x^{-2}$ & $\varepsilon x^{-1}$ & $\p_t,\, x\p_x,\, 2t\p_t-2\varepsilon tx\p_x+u\p_u,$ \\
&&&&& $e^{\varepsilon t}(x^2\p_x-xu\p_u)$\\
\refstepcounter{tbn}\label{A-43B0f1h}\thetbn a & $-4/3$ & 0 & 1 & 1 &
$\p_t,\, 4t\p_t+3u\p_u,\, \p_x,\, 2t\p_t+x\p_x,$ \\ &&&&& $x^2\p_x-3xu\p_u$ \\
\thetbn b & $-4/3$ & 1 & 1 & 1 & $\p_t,\, 4t\p_t+4x\p_x-3u\p_u,\, 2t\p_t+(x-t)\p_x,$ \\
 &&&&&$\p_x,\, (x+t)^2\p_x-3(x+t)u\p_u$ \\
\thetbn c & $-4/3$ & 1 & 1 & $\varepsilon x$ & $\p_t,\, 2x\p_x-3u\p_u,\, e^{-\varepsilon t}\p_x,$ \\
&&&&& $e^{-2\varepsilon t}(\p_t-\varepsilon x\p_x),\, e^{\varepsilon t}(x^2\p_x-3xu\p_u)$ \\
\refstepcounter{tbn}\label{A1Buf1h1}\thetbn & 0 & $u$ & 1 & 1 & $ \p_t,\, \p_x,\, t\p_x-\p_u,\, 2t\p_t+x\p_x-u\p_u,$ \\
&&&&& $t^2\p_t+tx\p_x-(tu+x)\p_u$ \\
%
\hline
\end{tabular}
\end{center}
{\footnotesize
Here $\nu\ne\mu$; $\varepsilon=\pm1$;
$\varkappa\in\{-1,0,1\}$ in case~\ref{AumBunkf1h1};
$q\ne-1$ in case~\ref{AumBunfxlhxg} (otherwise it is a subcase of case~1.\ref{gcAaBafexh1}a$'$);
$p\ne0$ in cases~\ref{AumBumlnuffhh} and~\ref{A1Buffhh} (otherwise they are subcases of case~1.\ref{gcAaBafexh1}a);
$p=\pm1$ in cases~\ref{A1Bafh2hh}, \ref{A1Bufh2hh} and~\ref{A1Blnufh2hh};
$p\ne-2,-(3\mu+4)/(\mu+1)$ in case~\ref{AumB0ff3hh}b;
$\alpha$, $\beta$, $\gamma_1$, $\gamma_0=\const$;
$\beta\ne0$ in case~\ref{Au-1B1fx-3ebxhx-2ebx} (otherwise it is a subcase of case~\ref{AumB0ff3hh}.b);
in case~\ref{AumB0ff3hh}a
$(\beta,\gamma_1,\gamma_0,\alpha)\in\{(\pm1,0,1,\hat\alpha),(1,1,0,\check\alpha),(0,1,0,\check\alpha)\}$
if $\mu=-1$ and
$(\beta,\gamma_1,\gamma_0,\alpha)\in\{(1,0,\pm1,\hat\alpha),(1,0,0,1),(0,1,0,\check\alpha),(0,0,1,1)\}$
if $\mu\ne-1$, where $\hat\alpha,\check\alpha=\const$, $\hat\alpha\geqslant0$;
\[
f^2(x)=\exp\left\{\int\frac{-(3\mu+4)\beta x-3}{(\mu+1)\beta x^2+x}\,dx\right\},\quad
f^3(x)=\exp\left\{\int\frac{-(3\mu+4)\beta x-2\gamma_1+\alpha}{(\mu+1)\beta x^2+\gamma_1x+\gamma_0}\,dx\right\}.
\]
Additional equivalence transformations:\\[1ex]
\setcounter{casetran}{0}
\refstepcounter{casetran}\thecasetran.
\ref{AumBunkf1h1}($\varkappa\ne0$) $\to$ \ref{AumBunkf1h1}($\varkappa=0$):\quad
$\tilde t=t$, $\tilde x=x+\varkappa t$, $\tilde u=u$;
\\[1ex]
\refstepcounter{casetran}\thecasetran.
\ref{AumB0ff3hh}b $\to$ \ref{AumB0ff3hh}a($\beta=\gamma_0=0$, $\alpha=(p+2)\gamma_1$), 
\ref{AumB0f1h}e $\to$ \ref{AumB0f1h}a($p=-\frac{3\mu+4}{\mu+1}$):\quad
$\tilde t=(e^{\varepsilon(p+2)t}-1)/(\varepsilon(p+2))$, $\tilde x=xe^{\varepsilon t}$, $\tilde u=u$;
\\[1ex]
\refstepcounter{casetran}\thecasetran. \ref{AumB0ff3hh}b${}^*$ $\to$ \ref{AumB0ff3hh}a($\beta=\gamma_1=0$, $\alpha=\gamma_0$):\quad
$\tilde t=e^{\varepsilon t}/\varepsilon$, $\tilde x=x+\varepsilon t$, $\tilde u=u$;
\\[1ex]
\refstepcounter{casetran}\thecasetran. \ref{AumB0ff3hh}c $\to$ \ref{AumB0ff3hh}a($\beta=\gamma_0=\alpha=0$),  
\ref{AumB0f1h}h $\to$ \ref{AumB0f1h}a:\quad
$\tilde t=t$, $\tilde x=xe^{\varepsilon t}$, $\tilde u=u$;
\\[1ex]
%
%
\refstepcounter{casetran}\thecasetran.
\ref{AumB0f1h}b $\to$ \ref{AumB0f1h}a, \ref{A-43B0f1h}b $\to$ \ref{A-43B0f1h}a:\quad
$\tilde t=t$, $\tilde x=x-t$, $\tilde u=u$;
\\[1ex]
\refstepcounter{casetran}\thecasetran.
\ref{AumB0f1h}c $\to$ \ref{AumB0f1h}a, \ref{A-43B0f1h}c $\to$ \ref{A-43B0f1h}a:\quad
$\tilde t=e^{2\varepsilon t}/(2\varepsilon)$, $\tilde x=xe^{\varepsilon t}$, $\tilde u=u$;
\\[1ex]
%
\refstepcounter{casetran}\thecasetran.
\ref{AumB0f1h}d $\to$ \ref{AumB0f1h}a: $\tilde t=t$, $\tilde x=-1/x$, $\tilde u=|x|^{-\frac1{1+\mu}}u$;
\\[1ex]
\refstepcounter{casetran}\thecasetran.
\ref{AumB0f1h}f $\to$ \ref{AumB0f1h}a($\mu=-1$): $\tilde t=t$, $\tilde x=x$, $\tilde u=e^{x}u$;
\\[1ex]
\refstepcounter{casetran}\thecasetran.
\ref{AumB0f1h}g $\to$ \ref{AumB0f1h}a($\mu=-1$):
  $\tilde t=e^{\varepsilon t}/\varepsilon$, $\tilde x=x+\varepsilon t$, $\tilde u=e^{x+\varepsilon t}u$.
%
%
\par
}

\vspace{3ex}

{\footnotesize\noindent
Additional equivalence transformations between cases from different tables:\\[1ex]
\setcounter{casetran}{0}
\refstepcounter{casetran}\thecasetran.
3.\ref{A1Bffhh} $\to$ 1.\ref{gcAaBafexh1}a$'$($A=1$, $p=-1$):\quad
$\tilde t=-e^{-t}\sign\ln|x|$, $\tilde x=e^{-t}\ln|x|$, $\tilde u=u$.
}

\vspace{2.5ex}

\setcounter{tbn}{0}
\begin{center}\footnotesize\renewcommand{\arraystretch}{1.2}
Table~1$'$. Case of $\forall A(u)$ (gauge $g=h$)\\[1ex]
\begin{tabular}{|l|c|c|c|l|}
\hline
N & $B(u)$ & $f(x)$ & $h(x)$ &\hfil Basis of A$^{\rm max}$ \\
\hline
\refstepcounter{tbn}\label{AaBafaghha}\thetbn & $\forall$ & $\forall$ & $\forall$ & $\p_t$ \\
\refstepcounter{tbn}\label{AaBafepxghh1}\thetbn a& $\forall$ & $e^{px}$ & 1 & $\p_t,\, pt\p_t+\p_x$ \\
\thetbn b& 1 & $h_x$ & $(hh_x)_x=h_x$ &
$\p_t,\, e^{-t}(\p_t-h\p_x)$ \\
\thetbn c& 1 & $h^{-1}|\hat G(h)|^p$ & $hh_x=G'(\hat G(h))$ &
$\p_t,\, e^{-(p+2)t}(\p_t-h\hat G(h)\p_x)$\\
\refstepcounter{tbn}\label{AaB1ffghhh}\thetbn & 1 & $|x|^{-1/2}e^{-|x|^{1/2}}$ &
$|x|^{-1/2}e^{|x|^{1/2}}$ & $\p_t,\, e^{-t/2}|x|^{1/2}\p_x$ \\
\refstepcounter{tbn}\label{AaB0f1ghh1}
\thetbn a& 0 & 1 & 1 & $\p_t,\, \p_x,\, 2t\p_t+x\p_x$ \\
\thetbn b& 1 & $|x|^{-1/2}$ & $|x|^{1/2}$ & $\p_t,\, e^{-t/2}|x|^{1/2}\p_x,\, e^{-t}(\p_t-x\p_x)$ \\
\thetbn c& 1 & 1 & 1 & $\p_t,\, \p_x,\, 2t\p_t+(x-t)\p_x$ \\
\hline
\end{tabular}
\end{center}
{\footnotesize
Here $p\ne-2$ in case~\ref{AaBafepxghh1}c, $G(z)=z|z|^p+\beta z$, $\beta=\const$,
$\hat G$ is the inverse function of~$G$.
}

\vspace{3ex}

\setcounter{tbn}{0}

\begin{center}\footnotesize\renewcommand{\arraystretch}{1.2}
Table~2$'$. Case of $A(u)=e^{\mu u}$ (gauge $g=h$)\\[1ex]
\begin{tabular}{|l|c|c|c|l|}
\hline
N & $B(u)$ & $f(x)$ & $h(x)$ &\hfil Basis of A$^{\rm max}$ \\
\hline
\refstepcounter{tbn}\label{AeuB0faghh1}\thetbn & 0 & $\forall$ & 1 & $\p_t,\, t\p_t-\p_u$ \\
\refstepcounter{tbn}\label{AeuBenufxpghhxq}\thetbn & $e^{\nu u}$ & $|x|^p$ & $|x|^q$ &
$\p_t,\, ((p-q+1)\mu-(p-q+2)\nu)t\p_t+(\mu-\nu)x\p_x+\p_u$ \\
\refstepcounter{tbn}\label{AeuBueufex2+xghhex2}\thetbn & $ue^u$ & $e^{p x^2+q x}$ & $e^{p x^2}$ &
$\p_t,\, (2p+q)t\p_t+\p_x-2p\p_u$ \\
\refstepcounter{tbn}\label{AeuBenu+kf1ghh1}\thetbn & $e^u+\varkappa$ & 1 & 1 &
$\p_t,\, \p_x,\, (\mu-2)t\p_t+((\mu-1)x+\varkappa t)\p_x+\p_u$ \\
\refstepcounter{tbn}\label{AeuBuf1ghh1}\thetbn & $u$ & 1 & 1 & $\p_t,\, \p_x,\, t\p_t+(x-t)\p_x+\p_u$ \\
\refstepcounter{tbn}\label{AeuB0ff1ghh1}\thetbn a & 0 & $f^1(x)$ & 1
&$\p_t,\, t\p_t-\p_u,\, (\beta x^2+\gamma_1x+\gamma_0)\p_x+(\beta x+\alpha)\p_u$ \\
\thetbn b& 1 & $|x|^p$ & $x|x|^p$
&$\p_t,\, x\p_x+\p_u,\, e^{-t}(\p_t-x\p_x)$ \\
\thetbn c& 1 & $e^{-x}$ & $e^{-x}$
&$\p_t,\, \p_x,\, t\p_t-t\p_x-\p_u$ \\
\refstepcounter{tbn}\label{AeuB0f1ghh1}\thetbn a & 0 & 1 & 1 & $\p_t,\, t\p_t-\p_u,\, 2t\p_t+x\p_x,\, \p_x$ \\
\thetbn b & 1 & 1 & 1 & $\p_t,\, \p_x,\, t\p_t-t\p_x-\p_u,\, 2t\p_t+(x-t)\p_x$ \\
\thetbn c& 1 & $|x|^{-1/2}$ & $|x|^{1/2}$
&$\p_t,\, x\p_x+\p_u,\, e^{-t/2}|x|^{1/2}\p_x,\, e^{-t}(\p_t-x\p_x)$ \\
\thetbn d & 0 & $x^{-3}$ & 1 & $\p_t,\, t\p_t-\p_u,\, x\p_x-\p_u,\, x^2\p_x+x\p_u$ \\
\thetbn e& 1 & $x$ & $x^2$
&$\p_t,\, x\p_x+\p_u,\, e^{-t}(\p_t-x\p_x),\, e^{-t}(\p_x-x^{-1}\p_u)$ \\
\hline
\end{tabular}
\end{center}
{\footnotesize
Here $(\mu,\,\nu)\in\{(0,\,1),\, (1,\,\nu)\}$, $\nu\ne\mu$ in case~\ref{AeuBenufxpghhxq};
$\mu\ne1$ in case~\ref{AeuBenu+kf1ghh1};
$\mu=1$ in cases~\ref{AeuB0faghh1}, \ref{AeuBuf1ghh1}--\ref{AeuB0f1ghh1}e;
$p\ne -\frac12,\, 1$ in case~\ref{AeuB0ff1ghh1}b.
\par}

\newpage

\setcounter{tbn}{0}

\begin{center}\footnotesize\renewcommand{\arraystretch}{1.35}
Table~3$'$. Case of $A(u)=|u|^\mu$ (gauge $g=h$)\\[1ex]
\begin{tabular}{|l|l|c|c|c|l|}
\hline
N & $\mu$ & $B(u)$ & $f(x)$ & $h(x)$ & \hfil Basis of A$^{\rm max}$ \\
\hline
\refstepcounter{tbn}\label{aumb0faghh1}\thetbn & $\forall$ & 0 & $\forall$ & 1 & $\p_t,\, \mu t\p_t-u\p_u$ \\
\refstepcounter{tbn}\label{aumbunfxpghhxq}\thetbn & $\forall$ & $|u|^\nu$ & $|x|^p$ & $|x|^q$ &
$\p_t,\, ((p-q+1)\mu-(p-q+2)\nu)t\p_t$\\ &&&&& ${}+(\mu-\nu)x\p_x+u\p_u$ \\
\refstepcounter{tbn}\label{aumbumlnufeax2+xghheax}\thetbn & $\forall$ & $|u|^\mu\ln|u|$ & $e^{p x^2+q x}$ &
$e^{p x^2}$ & $\p_t,\, (2\mu p+q)t\p_t+\p_x-2p u\p_u$ \\
\refstepcounter{tbn}\label{aumb1fghh}\thetbn & $\forall$ & 1 & $(\hat G(h_x))^{-1}$ & $h$ &
$\p_t,\, e^t(\p_t+{}$\\ &&&&& ${}+((2\mu+3)\beta \hat G(h_x)+2)h_x^{-1}\p_x-\beta xu\p_u)$ \\
\refstepcounter{tbn}\label{a1bafeax2ghheax2}\thetbn & 0 & $\forall$ & $e^{p x^2}$ &
$e^{p x^2}$ & $\p_t,\, e^{-2p t}\p_x$ \\
\refstepcounter{tbn}\label{a1bafexhghhgex}\thetbn & 0 & $\forall$ & $ e^{x+\gamma e^x}$ & $e^{\gamma e^x}$ &
$\p_t,\, e^{-\gamma t}(\p_t-\gamma \p_x)$ \\
\refstepcounter{tbn}\label{a1bufeax2+xghheax}\thetbn & 0 & $u$ & $e^{p x^2+x}$ &
$e^{p x^2}$ &$\p_t,\, t\p_t+\p_x-2p\p_u$ \\
\refstepcounter{tbn}\label{aumbun+kf1ghh1}\thetbn & $\forall$ & $|u|^\nu+\varkappa$ & 1 & 1 & $\p_t,\, \p_x,$ \\
&&&&& $(\mu-2\nu)t\p_t+((\mu-\nu)x+\nu\varkappa t)\p_x+u\p_u$ \\
\refstepcounter{tbn}\label{aumblnuf1ghh1}\thetbn & $\forall$ & $\ln|u|$ & 1 & 1 &
$\p_t,\, \p_x,\, \mu t\p_t+(\mu x-t)\p_x+u\p_u$ \\
\refstepcounter{tbn}\label{a1bufeax2ghheax2}\thetbn & 0 & $u$ & $e^{p x^2}$ & $e^{p x^2}$ &
$\p_t,\, e^{-2p t}\p_x,\, \p_x-2p\p_u$ \\
\refstepcounter{tbn}\label{a1blnufeax2ghheax2}\thetbn & 0 & $\ln|u|$ & $e^{p x^2}$ & $e^{p x^2}$ &
$\p_t,\, e^{-2p t}\p_x,\, \p_x-2p u\p_u$ \\
\refstepcounter{tbn}\label{aumb0ff3ghh1}\thetbn a & $\forall$ & 0 & $f^3(x)$ & 1 & $\p_t,\, \mu t\p_t-u\p_u,$ \\
&&&&& $\alpha t\p_t+((\mu+1)\beta x^2+\gamma_1x+\gamma_0)\p_x+\beta xu\p_u$\\
\thetbn b& $\forall$ & 1 & $|x|^p$ & $x|x|^p$ &
$\p_t,\, \mu x\p_x+u\p_u,\, e^{-t}(\p_t-x\p_x)$ \\
\thetbn c& $\ne-2$ & 1 & $e^{-x}$ & $e^{-x}$ &
$\p_t,\, \p_x,\, \mu t\p_t-\mu t\p_x-u\p_u$ \\
\refstepcounter{tbn}\label{au-65b1f1ghhx2/3}\thetbn & $-6/5$ & 1 & 1 & $x^{2/3}$ &
$\p_t,\, 2t\p_t+6x\p_x-5u\p_u,$ \\ &&&&& $t^2\p_t+(9x^{4/3}+6tx)\p_x-5(t+3x^{1/3})u\p_u$ \\
\refstepcounter{tbn}\label{aumb0f1ghh1}\thetbn a & $\ne -4/3$ & 0 & 1 & 1 &
$\p_t,\, \mu t\p_t-u\p_u,\, \p_x,\, 2t\p_t+x\p_x$ \\
\thetbn b & $\ne -4/3$ & 1 & 1& 1 & $\p_t,\, \mu t\p_t-\mu t\p_x-u\p_u,\, \p_x,\, 2t\p_t+(x-t)\p_x$ \\
\thetbn c& $\ne-4/3$&1 & $|x|^{-1/2}$ & $|x|^{1/2}$
&$\p_t,\, \mu x\p_x+\p_u,\, e^{-t/2}|x|^{1/2}\p_x,\, e^{-t}(\p_t-x\p_x)$ \\
\thetbn d & $\ne -4/3,-1$ & 0 & $|x|^{-\frac{3\mu+4}{\mu+1}}$ & 1 &
$\p_t,\, \mu t\p_t-u\p_u,\, (\mu+2)t\p_t-(\mu+1)x\p_x,$ \\ &&&&& $(\mu+1)x^2\p_x+xu\p_u$ \\
\thetbn e& $\ne-2,-4/3,-1$&1 & $|x|^{\frac{\mu+1}{\mu+2}}$ & $x|x|^{\frac{\mu+1}{\mu+2}}$
&$\p_t,\, \mu x\p_x+u\p_u,\, e^{-t}(\p_t-x\p_x), $ \\ &&&&& $
e^{-\frac{\mu+1}{\mu+2}}((\mu+2)x^{\frac1{\mu+2}}\p_x-x^{-\frac{\mu+1}{\mu+2}}u\p_u)$ \\
\thetbn f & $-1$ & 0 & $e^x$ & 1 & $\p_t,\, t\p_t+u\p_u,\, \p_x-u\p_u,\, 2t\p_t+x\p_x-xu\p_u$ \\
\thetbn g& $-1$ & 1 & 1 & $x$ &
$\p_t,\, x\p_x-u\p_u,\, e^{-t}(\p_t-x\p_x),$ \\ &&&&& $x(\ln x+t-2)\p_x-(\ln x+t)u\p_u$ \\
\thetbn h& $-2$ & 1 & $e^{-x}$ & $e^{-x}$ &
$\p_t,\,\p_x,\, 2t\p_t-2t\p_x+u\p_u,\, e^{t+x}(\p_x-u\p_u)$ \\
\refstepcounter{tbn}\label{au-43b0f1ghh1}\thetbn a & $-4/3$ & 0 & 1 & 1 &
$\p_t,\, 4t\p_t+3u\p_u,\, \p_x,\, 2t\p_t+x\p_x,$\\ &&&&& $x^2\p_x-3xu\p_u$ \\
\thetbn b & $-4/3$ & 1 & 1 & 1 & $\p_t,\, 4t\p_t+4x\p_x-3u\p_u,\, 2t\p_t+(x-t)\p_x,$ \\
&&&&& $\p_x,\, (x+t)^2\p_x-3(x+t)u\p_u$ \\
\thetbn c& $-4/3$&1 & $|x|^{-1/2}$ & $|x|^{1/2}$
&$\p_t,\, 4x\p_x-3u\p_u,\, e^{-t}(\p_t-x\p_x),$ \\
&&&&& $e^{-t/2}|x|^{1/2}\p_x,\,e^{t/2}(2x^{3/2}\p_x-3x^{1/2}u\p_u)$ \\
\refstepcounter{tbn}\label{a1buf1ghh1}\thetbn & 0 & $u$ & 1 & 1 &
$\p_t,\, \p_x,\, t\p_x-\p_u,\, 2t\p_t+x\p_x-u\p_u,$ \\ &&&&&$t^2\p_t+tx\p_x-(tu+x)\p_u$ \\
\hline
\end{tabular}
\end{center}
{\footnotesize
Here $\nu\ne\mu$.
In case~\ref{aumb1fghh} $G(z)=-\dfrac{(2\mu+3)\beta z+2}{(\mu+1)\beta z^2+z}$, $\beta=\const$,
$\hat G$ is the inverse function of~$G$ and $h$ satisfies the equation $h_{xx}h=G'(\hat G(h_x))$.
$p=\pm1$ in cases~\ref{a1bafeax2ghheax2}, \ref{a1bufeax2ghheax2} and~\ref{a1blnufeax2ghheax2};
$\gamma=\pm1$ in case~\ref{a1bafexhghhgex};
$\varkappa\in\{-1,0,1\}$ in case~\ref{aumbun+kf1ghh1}.
}

\newpage

\section{Proof of results of group classification}\label{SectionProofOfGrCl}

It seems impossible to formulate complete results of group classification of class~\eqref{eqDKfgh}
with respect to the usual equivalence group~$G^\sim$ in a closed form~\cite{Ivanova&Sophocleous2006}.
At the same time, it is quite easy to solve the problem of group classification with respect to
the extended equivalence group~$\hat G^\sim$.

Our method is based on the fact that the substitution of the coefficients of any operator from
$A^{\max} \backslash A^\cap$ into the classifying equations~\eqref{deteq2}--\eqref{deteq4}
results in nonidentity equations for arbitrary elements.

In the problem under consideration, the procedure of looking for the possible cases of extensions mostly
depends on equation~\eqref{deteq2}.
For any operator $Q\in A^{\max}$ the substitution of its coefficients into
equation~\eqref{deteq2} gives some equations on~$A$ of the general form
\begin{equation} \label{deq}
(\alpha u+\beta)A_u=\gamma A,
\end{equation}
where $\alpha$, $\beta$ and $\gamma$ are constants.
The set of coefficient triples $(\alpha,\beta,\gamma)$ collected for all operators from $A^{\max}$ is a linear space.
The dimension $k=k(A^{\max})$ of this space is not greater than 2
otherwise the corresponding equations form an incompatible system on $A$.
The value of $k$ is an invariant of the transformations from $\hat G^{\sim}$.
Therefore, there exist three $\hat G^{\sim}$-inequivalent cases for the value of $k$: $k=0$, $k=1$ and $k=2$.
We consider these possibilities separately (\emph{furcate split}), omitting cumbersome technical calculations.

\begin{note}
The choice of a gauge for the arbitrary elements is very important for solving and for the final presentation of results.
It is more convenient to constrain the parameter-function~$g$ instead of~$f$ in class~\eqref{eqDKfgh}.
The next problem is the choice between different gauges of~$g$.
The case $B\not\in\langle 1,A\rangle$ and $k\geqslant1$ is easier to be investigated in the gauge $g=h$.
In the other cases we obtain results in a simpler explicit form and in an easier way using the gauge $g=1$.
\end{note}

\noindent{\mathversion{bold}$k=0$} (the gauge $g=1$, table~1).
Then the coefficients of any operator from $A^{\max}$ are to satisfy the system
\begin{equation}\label{dse_dall}
\eta=0,\quad 2\xi_x-\tau_t+\frac{f_x}f\xi=0,\quad f\xi_t=A\xi_{xx}-B(h\xi)_x.
\end{equation}

Let us suppose that $B\notin \langle 1,\, A \rangle$.
It follows from the last equation of the system~(\ref{dse_dall}) that up to $\hat G^{\sim}$-equivalence $\xi_x=\xi_t=0$.
Therefore, the second equation is a nonidentity equation for $f$ of the form $f_x=\mu f$ without fail.
Solving this equation yields case~2a.

Now let $B\in \langle 1,\, A \rangle$, i.e., $B=\delta\bmod \hat G^{\sim}_1$, where
$\delta\in\{0,1\}$.
Then the last equation of~(\ref{dse_dall}) can be decomposed into the equations
$\xi_{xx}=0$, $\delta(h\xi)_x+f\xi_t=0$.
Integrating the latter equations up to $\hat G^{\sim}_1$-equivalence results in cases~2b--4c of table~1.

\medskip

\noindent{\mathversion{bold}$k=1$} (the gauges $g=1$ and $g=h$, tables~2,\,3 and 2$'$,\,3$'$).
Then $A\in\{u^{\mu},\mu\ne0,e^u\}\bmod G^{\Equiv}_1$,
and we can assume that there exists $Q\in A^{\max}$ with $\eta\ne 0$,
otherwise there is no additional extension of the maximal Lie invariance algebra in comparison with the case $k=0$.
Below we consider the exponential and power cases of~$A$ simultaneously and
write down the differences of the case $A=e^u$ with the one $A=u^\mu$ in brackets.
If $A=e^u$, we assume $\mu=1$.

Equations~\eqref{deteq1} and~\eqref{deteq2} imply that $\eta=\zeta(t,x)u$ ($\eta=\zeta(t,x)$).
Then equation~\eqref{deteq4} with respect to $B$ looks like
\[uB_u=\nu B+\lambda A+\varkappa\quad (B_u=\nu B+\lambda A+\varkappa)\] where $\nu$, $\lambda$ and $\varkappa$ are constants, otherwise $\eta\equiv0$.

Consider first the case $B\not\in\langle 1,A\rangle$ using the gauge $g=h$.
Under the above suppositions, equations~\eqref{deteq1}--\eqref{deteq4} can be rewritten as
\begin{gather*}
\frac{\varphi_x}\varphi\xi=(2\nu-\mu)\zeta+\tau_t,\quad \zeta_t=\zeta_x=0,\quad
\xi_{xx}=\xi_{tx}=\xi_{tt}=\tau_{tt}=0,\\
\xi_x=(\mu-\nu)\zeta,\quad \left(\xi\frac{h_x}h\right)_x=-\lambda\zeta,\quad \varphi\xi_t=-\varkappa\zeta.
\end{gather*}
Here and below $\varphi=f/h$.
(Note, that the gauge $g=1$ leads to the determining equations that cannot be integrated explicitly. See, e.g., case~2.\ref{AeuBueufh}.)

If $\varkappa\ne0$ then there exists $Q\in A^{\max}$ such that $\xi_t\ne0$.
Therefore, $\varphi_x=0$, i.e., $\varphi=1\bmod G^{\sim}_h$ and $\tau=(\mu-2\nu)\zeta t+c_0$, $\xi=(\mu-\nu)\zeta x-\varkappa\zeta t+c_1$.
Hence $(h_x/h)_x=0$, i.e., $h=h_0e^{h_1x}=1\bmod \hat G^{\sim}_h$ and then $f=g=h=1$, $uB_u=\nu B+\varkappa$ ($B_u=\nu B+\varkappa$).
Solving this equation up to $G^{\sim}_h$ yields $B=|u|^{\nu}-\varkappa/\nu$ ($B=e^{\nu u}-\varkappa/\nu$) if $\nu\ne0$ and
$B=\ln|u|$ ($B=u$) if $\nu=0$.
Scaling the value of an arbitrary constant $\varkappa$,
we obtain cases~3$'$.\ref{aumbun+kf1ghh1} and~3$'$.\ref{aumblnuf1ghh1}
(2$'$.\ref{AeuBenu+kf1ghh1} and~2$'$.\ref{AeuBuf1ghh1}).

Now let $\varkappa=0$. Then $\xi_t=0$ and $\varphi\in\{e^x,\,|x|^r, r\ne0,\,1\}\bmod G^{\sim}_h$.
For $\varphi=e^x$ the determining equations implies that $\xi_x=0$, $\nu=\mu$, $\lambda\ne0$, $(h_x/h)_x=2\alpha$.
Therefore, $h=h_0e^{\alpha x^2+h_1x}=e^{\alpha x^2}\bmod G^{\sim}_h$, $\alpha\ne0$, $f=e^{\alpha x^2+x}$ and
$B=\lambda|u|^\mu\ln|u|\bmod G^{\sim}_h$ ($B=\lambda ue^u\bmod G^{\sim}_h$)
that falls, after rescaling $x$, precisely into case~3$'$.\ref{aumbumlnufeax2+xghheax} (2$'$.\ref{AeuBueufex2+xghhex2}).

If $\varphi=|x|^r$ with $r\ne 0$ then $r\xi/x=(2\nu-\mu)\zeta+\tau_t$. Hence, $\xi=(\mu-\nu)\zeta x$,
$\tau_t=((r+1)\mu-(r+2)\nu)\zeta$ and $(\mu-\nu)(xh_x/h)_x=-\lambda$. Since $\mu\ne\nu$ (otherwise, $B\in\langle1,A\rangle$)
we have $\lambda=0\bmod \hat G^{\sim}_h$. Therefore, $h=|x|^q\bmod \hat G^{\sim}$. Then $f=|x|^p$, $p\ne q$,
and we obtain case~3$'$.\ref{AumBunfxlhxg} (2$'$.\ref{AeuBenufxpghhxq}).

The value $\varphi=1$ results in $\tau_t=(\mu-2\nu)\zeta$, $\xi=(\mu-\nu)\zeta x+\xi_0$.
If $\nu=\mu$ then $\lambda\ne0$ (otherwise, $B\in\langle1,A\rangle$), $\lambda=1\bmod G^{\sim}_h$,
$(h_x/h)_x=2\alpha$. Therefore, $h=h_0e^{\alpha x^2+h_1x}=e^{\alpha x^2}\bmod G^{\sim}_h$,
$B=|u|^\mu\ln|u|$ ($B=ue^u$) that follows to case~3$'$.\ref{aumbumlnufeax2+xghheax} (2$'$.\ref{AeuBueufex2+xghhex2}).
If $\nu\ne\mu$, then $\lambda=0\bmod G^{\sim}_h$. Therefore, $h\in\{|x|^q,\, 1,\, e^x\}\bmod G^{\sim}_h$
that yields special subcases of~3$'$.\ref{aumbunfxpghhxq}, 3$'$.\ref{aumbun+kf1ghh1} and~1.\ref{gcAaBafexh1}a.
(2$'$.\ref{AeuBenufxpghhxq},~2$'$.\ref{AeuBenu+kf1ghh1} and~1.\ref{gcAaBafexh1}a), respectively.

All the remaining cases are investigated similarly to the above one. We present only the main steps of the integration procedure.

In contrast to the previous case, it is more convenient to study the case $B\in\langle 1,\, A\rangle$ and $A\ne$constant using the gauge $g=1$.
In this case $B\in\{0,1\}\bmod\hat G^{\sim}_1$ and  the determining equations are reduced to the system
\begin{gather*}
2\xi_x+\dfrac{f_x}f\xi=\mu\zeta+\tau_t,\quad \zeta_{xx}=0,\quad B\zeta_x=\varphi\zeta_t,\\
\left(\xi_x+\dfrac{\varphi_x}\varphi\xi-\tau_t\right)B=\varphi\xi_t,\quad \xi_{xx}=2(\zeta_x\mu+\eta^1_x),
\end{gather*}
where $\eta^1=\zeta$ ($\eta^1=0$). Therefore, $\zeta=\zeta^1(t)x+\zeta^0(t)$ and
$\xi=\xi^2(t)x^2+\xi^1(t)x+\xi^0(t)$, where $\xi^2(t)=\mu\zeta^1+\eta^{11}$ and $\eta^{11}=\zeta^1(t)$ ($\eta^{11}=0$).
Plugging these values to the first determining equation, we obtain
\[
(\xi^2x^2+\xi^1x+\xi^0)\dfrac{f_x}f=-(3\mu\zeta^1+4\eta^{11})x+\mu\zeta^0+\tau_t-2\xi^1.
\]
This condition gives equations of the form $(\alpha_2x^2+\alpha_1x+\alpha_0)f_x=(\beta_1x+\beta_0)f$ for $f$,
whose coefficient tuples collected for all operators from $A^{\max}$ form a linear space.
Denote by $l$ the dimension of the space.
If $l=0$ then  $\xi=0$, $\zeta^1=0$ and $\tau_t=-\mu\zeta_0$. Considering the case $l=1$,
we get $(\alpha_2,\beta_1)\ne(0,0)$ and $(\alpha_0,\beta_0)\ne(0,0)$ otherwise $l>1$.
At last, the condition $l\geqslant2$ implies $f\in\{1,e^{px},|x|^p, p\ne0\}\bmod G^{\sim}_1$.

The direct substitution of the above values into the determining equations for $B=0$ and obvious integration
lead to the cases~3$'$.\ref{aumb0faghh1} (2$'$.\ref{AeuB0faghh1}) (case $l=0$), 3$'$.\ref{aumb0ff3ghh1}a (2$'$.\ref{AeuB0ff1ghh1}a) (case $l=1$)
and~3$'$.\ref{aumb0f1ghh1}a (2$'$.\ref{AeuB0f1ghh1}a), 3$'$.\ref{au-43b0f1ghh1}, 3$'$.\ref{aumb0f1ghh1}f,
3$'$.\ref{aumb0f1ghh1}d (2$'$.\ref{AeuB0f1ghh1}d) (case $l=2$).

The classification in the case $B=1$ is more cumbersome. First we show that if $l=0$ then $\zeta=0$ and therefore,
$\eta=0$ and $A^{\max}=A^\cap$. To complete the consideration of this case,
one should classify separately three essentially different cases: $\varphi$ is arbitrary ($\varphi\ne1,\varepsilon/x\bmod G^{\sim}_1$),
$\varphi=1$ and $\varphi=\varepsilon/x$. Each of these cases can be studied in a way that is similar to the above consideration.

\medskip
\noindent\mathversion{bold}\mbox{$k=2$}\mathversion{normal} (the gauge $g=h$, tables~2$'$ and $3'$).
The assumption on two independent equations of form~\eqref{deteq2} on~$A$ yields
$A=\const$, i.e., $A=1\bmod G^{\Equiv}_h$. Consider the case $B_u\not=0$ (otherwise, equation~\eqref{eqDKfh} is linear).
The most suitable gauge here is $g=h$.
Equations~\eqref{deteq1}--\eqref{deteq4} can be rewritten as
\begin{gather*}
2\xi_x-\tau_t+\left(\dfrac{f_x}f-\dfrac{h_x}h\right)\xi=0, \quad (h\eta_{x})_x+Bh\eta_x-f\eta_t=0,\\
\eta B_u+\xi_xB+\varphi \xi_t+\xi\dfrac{h_{xx}}{h}-\xi_{xx}+2\eta^1_x+\left(\dfrac\xi h\right)_xh_x=0.
\end{gather*}
The latter equation looks as $(\alpha u+\beta)B_u=\gamma B+\delta$ with respect to $B$,
where $\alpha,\beta,\gamma,\delta=\const$. Therefore, up to $\hat G^{\Equiv}_h$-equivalence
$B$ is to take one of the values:
\[
B - \forall;\quad
B=u^\nu,\;\nu\ne0,1;\quad
B=\ln u;\quad
B=e^u;\quad
B=u.
\]
Classification for these values is carried out in the way similar to the above.
The derived extensions are entered in either table~2$'$ or table~3$'$.

\section{Admissible transformations}\label{SectionAdmisTransf}

The presence of the nontrivial extended equivalence group and many additional equivalence transformations
indicates that the set of all admissible transformations of class~\eqref{eqDKfgh} has a complicated structure.
In this section we describe only basic properties of admissible transformations of class~\eqref{eqDKfgh},
which are useful for finding additional equivalence transformations.
In fact, the problems of finding of all possible admissible transformations are very difficult to solve even for
classes of simpler structure.
See, e.g.,~\cite{Ivanova&Popovych&Eshraghi2005GammaNormalizedSerbia,
Kingston&Sophocleous1991,Kingston&Sophocleous1998,Kingston&Sophocleous2001,
Pallikaros&Sophocleous1995,Popovych2006,Popovych&Kunzinger&Eshraghi2006,
Vaneeva&Johnpillai&Popovych&Sophocleous2006,Vaneeva&Popovych&Sophocleous2009}.

Any point transformation in the space of the variables $(t,x,u)$ has the form
\[
\tilde t=T(t,x,u),\quad \tilde x=X(t,x,u),\quad \tilde u=U(t,x,u)
\]
where the nonsingularity condition $J=\partial(T,X,U)/\partial(t,x,u)\ne 0$ is satisfied.
In what follows tilde (resp.\ non-tilde) arbitrary elements depend on tilde (resp.\ non-tilde) variables.

It is well known (see, e.g., \cite{Kingston&Sophocleous1998}) that
for any point transformation between two evolutionary equations of order~$n$ greater than 1
(i.e., equations of the form $u_t=H(t,x,u,u_1,\ldots,u_n)$ where $u_k=\p^ku/\p x^k$, $k=1,2,\dots$, $H_{u_n}\ne 0$)
the component corresponding to the variable~$t$ depends only on $t$. That is, $\tilde t=T(t)$.
The right hand sides $H$ and~$\tilde H$ of the initial and transformed equations are related by the formula
\begin{equation}\label{(i)}
(X_xU_u-X_uU_x)H=(X_x+u_xX_u)T_t\tilde H+X_t(U_x+u_xU_u)-(X_x+u_xX_u)U_t.
\end{equation}

\begin{lemma}\label{LemmaOnAdmTransOfQuasiLin1ndOrderEvolEqs}
Any point transformation between two evolutionary second-order quasi-linear equations having the form
$u_t=F(t,x,u)u_{xx}+G(t,x,u,u_x)$  where $F\ne 0$ is projectible, i.e.,
$\tilde t =T(t),$ $\tilde x =X(t,x),$ $\tilde u =U(t,x,u).$
\end{lemma}
\begin{proof}
We set in \eqref{(i)} $H=F(t,x,u)u_{xx}+G(t,x,u,u_x)$ and
$\tilde H=\tilde F(\tilde t,\tilde x,\tilde u)\tilde u_{\tilde x\tilde x}
+\tilde G(\tilde t,\tilde x,\tilde u,\tilde u_{\tilde x})$.
Note that
$u_{\tilde x}=V:=(D_xX)^{-1}D_xU$, $\tilde u_{\tilde x\tilde x}=(D_xX)^{-1}D_xV$,
where $D_x$ stands for the total derivative with respect to the variable~$x$,
$D_x=\p_x+u_x\p_u+u_{xx}\p_{u_x}+\cdots$.
Collecting coefficients of $u_{xx}$ in the simplified \eqref{(i)} gives
$(X_x+u_xX_u)^2F=T_t\tilde F$.
Splitting the last equations with respect to~$u_x$ implies that $X_u=0$.
\end{proof}

We prove the following lemmas for the particular case of class~\eqref{eqDKfgh}
although there exist similar statements for more general classes of evolutionary equations.
Using Lemma~\ref{LemmaOnAdmTransOfQuasiLin1ndOrderEvolEqs} and the representations of~$H$ and $\tilde H$
for equations from class~\eqref{eqDKfgh}, we have that $J=T_tX_xU_u\ne0$ and
\begin{gather}
(gAu_{xx}+gA_uu_x^2+g_xAu_x+hBu_x)\frac{U_u}{f}=
\nonumber\\
(\tilde g\tilde A\tilde u_{\tilde x\tilde x}+\tilde g\tilde A_{\tilde u}\tilde u_{\tilde x}^2+
\tilde g_{\tilde x}\tilde A\tilde u_{\tilde x}+\tilde h\tilde B\tilde u_{\tilde x})
\frac{T_t}{\tilde f}
+\frac{X_t}{X_x}(U_x+u_xU_u)-U_t.\label{(iii)}
\end{gather}

\begin{lemma}
Any point transformation between two equations from class~\eqref{eqDKfgh} is linear with respect to $u$:
$\tilde t =T(t),$ $\tilde x =X(t,x),$ $\tilde u =U^1(t,x)u+U^0(t,x)$, where $T_tX_xU^1\ne0$.
\end{lemma}

\begin{proof}
Collecting coefficients of $u_{xx}$ and $u_x^2$ in~\eqref{(iii)} respectively gives
\begin{gather}\label{(iv)}
A=K\tilde A, \quad
A_u=\frac K{U_u}(U_{uu}\tilde A+U_u^2 \tilde A_{\tilde u}), \quad
K:=\frac{T_t}{X_x^2}\frac{f}{g}\frac{\tilde g}{\tilde f}.
\end{gather}
We differentiate the first equation of~\eqref{(iv)} with respect to $u$ and subtract from the second one.
As a result, we obtain that $KU_{uu}\tilde A/U_u=0$.
Hence $U_{uu}=0$.
\end{proof}

\begin{lemma}
Modulo $G^{\Equiv}$, there exist no point transformations changing the coefficient~$A$.
\end{lemma}

\begin{proof}
Since $T_u=X_u=U_{uu}=0$, the first equation of~\eqref{(iv)} implies that
the arbitrary elements~$A$ and~$\tilde A$ are related by the formula
$\varepsilon_2A(u)=\tilde A(\delta_3 u+\delta_4)$,
where $\varepsilon_2$, $\delta_3$ and  $\delta_4$ are constants, $\varepsilon_2\delta_3\ne0$.
Such transformation of~$A$ can be realized via a usual equivalence transformation
(cf.\ Theorem~\ref{TheoremOnUsualEquivGroupOfEqDKfgh}).
\end{proof}

\begin{lemma}
$(U_t,U_x)\ne(0,0)$ for a point transformation between two equations from class~\eqref{eqDKfgh}
only if $A\in\{u^\mu,\;e^u\}\bmod G^{\Equiv}$.
\end{lemma}

\begin{proof}
Differentiating the first equation of~\eqref{(iv)} with respect to $t$ and $x$, we obtain
\[
K_t\tilde A+K(U^1_tu+U^0_t)\tilde A_{\tilde u}=0,\quad
K_x\tilde A+K(U^1_xu+U^0_x)\tilde A_{\tilde u}=0.
\]
If $(U_t,U_x)\ne(0,0)$, this differential consequences implies that
the arbitrary element $\tilde A$ necessarily satisfies an ordinary differential equation of the form
$(\lambda_1\tilde u+\lambda_2)\tilde A_{\tilde u}+\lambda_3\tilde A=0$,
where $(\lambda_1,\lambda_2)\ne0$.
Solving this ordinary differential equation and using the $G^{\Equiv}$-equivalence,
we conclude that either $\tilde A=\tilde u^\mu$ or $\tilde A=e^{\tilde u}$ modulo $G^{\Equiv}$.
Therefore, also $A\in\{u^\mu,\;e^u\}\bmod G^{\Equiv}$.
\end{proof}

\begin{lemma}
$U^0=0$ if $A=u^\mu$ and $\tilde A=\tilde u^\mu$, where $\mu\ne0$.
$U^1=1$ if $A=e^u$ and $\tilde A=e^{\tilde u}$.
\end{lemma}

The proof directly follows from the first equation of~\eqref{(iv)}.
The proof of the next lemma is more complicated and will be presented in the second part of the series.

\begin{lemma}
The arbitrary elements~$B$ and~$\tilde B$ of similar equations from class~\eqref{eqDKfgh} are related by the formula
$\varepsilon_3(B(u)+\varepsilon_4A(u)+\varepsilon_5)=\tilde B(\delta_3 u+\delta_4)$,

\end{lemma}

\section{Conclusion}\label{SectionConclusion}

The present paper is the first part of a series of works on investigation
of variable-coefficient diffusion--convection equations~\eqref{eqDKfgh}
in the framework of modern group analysis of differential equations.
After discussing precise definitions, the algorithm of group classification and its modifications
in the general case, we have carried out the extended group classification of class~\eqref{eqDKfgh}.
The success in the classification and the clear presentation of the final result
have been achieved by the methodical applications of four original tools.
Namely, involving the equivalence relation with respect to the extended equivalence group instead of the usual one,
the choice of appropriate gauges, furcate split and
systematic usage of additional equivalences, based on preliminary description of admissible transformations.
Two of them (the extended equivalence group and appropriate gauges) are of crucial importance for
obtaining a closed and explicit classification list.
The extended equivalence group of class~\eqref{eqDKfgh} is the extension of the usual one
with the non-trivial group of gauge equivalence transformations including transformations which
are nonlocal in arbitrary elements.
Neglecting these transformations leads to critical swelling and complication of both calculations and results.
Moreover, under the presence of nonlocal gauge transformations, the group classification with respect to
the usual equivalence group cannot be done in the best way, since the corresponding classification list
necessarily contains different cases which are in fact associated with the same equation.
Class~\eqref{eqDKfgh} has been exhaustively classified for the two ``best'' gauges $g=1$ and $g=h$
and then the comparative analysis of them has been comprehensively made.
However, for future applications it is more suitable to use the variable gauge
when the value of $g$ (1 or $h$) depends on values of other arbitrary elements.
The classification under the variable gauge will be presented
in the second part~\cite{Ivanova&Popovych&Sophocleous2006Part2}.

Results adduced above can be developed to several directions.
In particular, the set of admissible transformations in class~\eqref{eqDKfgh} could be
studied more profoundly or even exhaustively described.
The group classification can be used for the construction of exact invariant solutions
of equations from class~\eqref{eqDKfgh} in different ways.
One of them involves additional equivalences.
In~\cite{Popovych&Ivanova2004NVCDCEs} known exact solutions of ``constant coefficient''
diffusion--convection equations were transformed by additional equivalence transformations
to new solutions of equations~\eqref{eqDKfh} with $h=1$ and complicated values of the parameter--function~$f$.
The same trick can be used for the entire class~\eqref{eqDKfgh} since
in view of Theorem~\ref{TheoremOnReducibilityOfEqsWith4DimAlgs}
all equations from this class, possessing four-dimensional Lie invariance algebras,
are reduced by point transformations to constant coefficient ones.
Another way is the standard method of Lie reduction.
In the second part of the series we will investigate
the unique ${\rm sl}(2,\mathbb{R})$-invariant equation~3.\ref{A-65B1fx2hx2}
which is ``essentially variable coefficient'' in the sense that it is not reducible to equations of
form~\eqref{eqDKfgh} with constant values of $f$, $g$ and $h$.

Analyzing Lie symmetries of class~\eqref{eqDKfgh} implies a number of interesting conjectures.
For example, the cases of exponential nonlinearities collected in table~2 can be regarded as limits of
cases with power nonlinearities from table~3.
This observation leads to the important and useful notion of contractions of equations and symmetries,
which also will be introduced, developed and applied in~\cite{Ivanova&Popovych&Sophocleous2006Part2}.
This notion will be naturally generalized to the notion of contractions of conservation laws
in the third part~\cite{Ivanova&Popovych&Sophocleous2006Part3} of the series,
where we will study the local and potential conservation laws of equations~\eqref{eqDKfh}.

\subsection*{Acknowledgements}

Research of NMI was supported by
the Cyprus Research Promotion Foundation (project number $\Pi$PO$\Sigma$E$\Lambda$KY$\Sigma$H/$\Pi$PONE/0308/01),
the Erwin Schr\"odinger Institute for Mathematical Physics (Vienna, Austria) in form of Junior Fellowship
and by the grant of the President of Ukraine for young scientists (project number GP/F11/0061).
Research of ROP was supported by Austrian Science Fund (FWF), Lise Meitner project M923-N13 and project~P20632.
ROP expresses his gratitude to the hospitality shown by University of Cyprus
during his visits to the University.

\end{document}